\documentclass[11pt]{article}
\usepackage[utf8]{inputenc} 
\usepackage{parskip}
\usepackage{amssymb, amsmath, amsthm, graphicx, subfigure}
\usepackage{enumerate}
\usepackage[dvipsnames]{xcolor}
\usepackage[colorlinks=true,linktoc=page]{hyperref}
\usepackage{braket}
\usepackage{esvect}
\usepackage[margin=1in]{geometry}
\usepackage{mathtools}
\usepackage{tikz}
\usepackage{algorithm}
\usepackage{algorithmicx}
\usepackage[noend]{algpseudocode}

\usepackage{mdframed}

\mathtoolsset{showonlyrefs=true}

\usetikzlibrary{patterns,patterns.meta}

\theoremstyle{plain}
\newtheorem{theorem}{Theorem}[section]
\newtheorem{lemma}[theorem]{Lemma}
\newtheorem{corollary}[theorem]{Corollary}

\newtheorem{fact}[theorem]{Fact}

\theoremstyle{definition}
\newtheorem{definition}[theorem]{Definition}

\newtheorem{conjecture}{Conjecture}

\theoremstyle{remark}

\newtheorem{example}[theorem]{Example}

\newcommand{\tail}{{\Bar{\Phi}}}

\newcommand{\E}{\mathbb{E}}

\newcommand{\br}{\mathbf{r}}
\newcommand{\bu}{\mathbf{u}}
\newcommand{\bv}{\mathbf{v}}

\newcommand{\bx}{\mathbf{x}}

\newcommand{\eps}{\varepsilon}

\newcommand{\Red}{\mathsf{R}}
\newcommand{\Green}{\mathsf{G}}
\newcommand{\Blue}{\mathsf{B}}
\newcommand{\Colour}{\mathsf{C}}
\newcommand{\Colours}{{\{\Red, \Green, \Blue\}}}
\newcommand{\col}{\mathrm{col}}

\newcommand{\KMS}{{\textup{\textsc{KMS}}}}
\newcommand{\KMSp}{{\textup{\textsc{KMS}}'}}

\title{Improved SDP-Based Algorithm for Coloring 3-Colorable Graphs}

\author{Nikhil Bansal\thanks{University of Michigan. Supported by NSF awards CCF-2327011 and CCF-2504995. Email: \texttt{bansaln@umich.edu}} \and Neng Huang\thanks{TU Dresden. Work was mostly done when NH was affiliated with the University of Michigan. NH has been funded by the European Union (Project POCOCOP, ERC Synergy Grant 101071674). Views and opinions expressed are however those of the author(s) only and do not necessarily reflect those of the European Union or the European Research Council Executive Agency. Neither the European Union nor the granting authority can be held responsible for them. Email: \texttt{neng.huang@tu-dresden.de}} \and Euiwoong Lee\thanks{University of Michigan. Supported in part by NSF award CCF-2236669. Email: \texttt{euiwoong@umich.edu}}}

\begin{document}
\maketitle

\begin{abstract}
    We present a polynomial-time algorithm that colors any 3-colorable $n$-vertex graph using $O(n^{0.19539})$ colors, improving upon the previous best bound of $\widetilde{O}(n^{0.19747})$ by Kawarabayashi, Thorup, and Yoneda [STOC 2024]. Our result constitutes the first progress in nearly two decades on SDP-based approaches to this problem.

    The earlier SDP-based algorithms of Arora, Chlamtáč, and Charikar [STOC 2006] and Chlamtáč [FOCS 2007] rely on extracting a large independent set from a suitably “random-looking” second-level neighborhood, under the assumption that the KMS algorithm [Karger, Motwani, and Sudan, JACM 1998] fails to find one globally. We extend their analysis to third-level neighborhoods. We then come up with a new vector $5/2$-coloring, which allows us to extract a large independent set from some third-level neighborhood. The new vector coloring construction may be of independent interest.
\end{abstract}
\setcounter{tocdepth}{1}
 {\small
    \thispagestyle{empty}
 }

\thispagestyle{empty}
\setcounter{page}{0}

\newpage

\allowdisplaybreaks
\section{Introduction}
We study the problem of coloring a 3-colorable graph with the fewest number of colors in polynomial time. It is one of the most classical tasks in theoretical computer science that has been extensively studied in both algorithm design and complexity theory for decades, but there is still a wide gap in our understanding.

On the hardness side, coloring 3-colorable graphs with three colors was proved to be NP-hard in 1974~\cite{garey1974some}. Ruling out four colors took place in the 90s~\cite{khanna2000hardness, guruswami2004hardness}. Using heavy algebraic machinery from Promise CSPs, Barto, Bulín, Krokhin, and Opršal~\cite{barto2021algebraic} proved that coloring 3-colorable graphs using 5 colors is NP-hard, which remains the best known hardness result assuming only $\mathbf{P} \neq \mathbf{NP}$.
Assuming variants of the Unique Games Conjecture, NP-hardness was established for any constant number of colors greater than three~\cite{dinur_conditional_2009, guruswami_d--1_2020, braverman2022invariance}. 

On the algorithmic front, there have been nice parallel developments from combinatorial methods and semidefinite programming (SDP) methods. 
From the combinatorial side, {\em Wigderson's trick} gives a simple algorithm that colors any $3$-colorable graph using $O(\sqrt{n})$ colors~\cite{wigderson_improving_1983}, which was improved to $O(\sqrt{n / \log n})$ by Berger and Rompel~\cite{berger1990better}. Blum introduced new combinatorial ideas to improve it further to $\widetilde{O}(n^{3/8})$ colors\footnote{In this paper, we use $\widetilde{\Theta}$, $\widetilde{\Omega}$, $\widetilde{O}$ to suppress (poly-)logarithmic factors in $n$.}~\cite{blum_new_1994}. 

SDP-based approaches started to play a crucial role from the 1990s, when Karger, Motwani, and Sudan gave an algorithm that uses $\widetilde{O}(n^{1/4})$ colors using a simple SDP algorithm combined with Wigderson's trick~\cite{karger_approximate_1998}. Blum and Karger~\cite{blum_tildeon314-coloring_1997} improved it to 
$\widetilde{O}(n^{3/14}) = \widetilde{O}(n^{0.2143})$ by combining the SDP method with Blum's combinatorial ideas.
Based on more sophisticated ideas for rounding SDPs in the 2000s, including the Arora-Rao-Vazirani algorithm for Sparsest Cut~\cite{arora2009expander} and the introduction of the Sum-of-Squares (SoS) hierarchy~\cite{parrilo2000structured, lasserre2001global}, 
Arora, Chlamtáč, and Charikar~\cite{arora_new_2006}, and 
Chlamtáč~\cite{chlamtac_approximation_2007} made further progress to $\widetilde{O}(n^{0.2111})$ and $\widetilde{O}(n^{0.2072})$ respectively.  

More recently, Kawarabayashi and Thorup, the first since Blum to suggest new combinatorial ideas, pushed the bound to $\widetilde{O}(n^{0.2049})$~\cite{kawarabayashi2012combinatorial}, and $\widetilde{O}(n^{0.19996})$~\cite {kawarabayashi_coloring_2017}, and later with Yoneda~\cite{kawarabayashi_better_2024}, they currently hold the best upper bound of $\widetilde{O}(n^{0.19747})$.

\subsection{Overview of Our Result}\label{sec:overview_intro}

 Our main contribution in this paper is achieving the first progress on the SDP-based approach since Chlamtáč~\cite{chlamtac_approximation_2007}, resulting in the following theorem. 

\begin{theorem}\label{theorem:main}
    There is a polynomial-time algorithm which, given a $3$-colorable graph $G$, produces a proper coloring with $O(n^{0.19539})$ colors.
\end{theorem}

We now present a high-level overview of our result. To better do so, we first give an overview of previous SDP algorithms and their combination with combinatorial algorithms. As mentioned earlier, the current best approach for coloring 3-colorable graphs uses a two-pronged attack, namely combinatorial algorithms which work well on dense graphs, and SDP-based algorithms which work well on sparse graphs. It is known that in order to obtain an $f(n)$-coloring, it is sufficient to find some $\Delta = \Delta(n)$ and make progress towards an $f(n)$-coloring separately for two families of graphs: those whose maximum degree is at most $\Delta$ (the sparse case), and those whose minimum degree is at least $\Delta$ (the dense case). (See Section~\ref{subsec:combine} for a formal definition of making progress.) On the combinatorial side, the current best result is the following theorem. 
\begin{theorem}[\cite{kawarabayashi_better_2024}]\label{thm:KTY}
    Given a 3-colorable $n$-vertex graph $G$ with minimum degree $\Delta > n^{1/2}$, we can make progress towards a $k$-coloring in polynomial time for some $k = n^{o(1)}\sqrt{n/{\Delta}}$.
\end{theorem}

On the SDP side, the only kind of progress (towards an $f(n)$-coloring) that SDP-based algorithms make is to find an $\Omega(n / f(n))$-sized independent set. This line of work started with Karger, Motwani, and Sudan's now textbook algorithm: given a 3-colorable graph $G = (V, E)$, it first solves a simple SDP relaxation\footnote{For simplicity, we will ignore precision issues and assume that SDPs appearing in this paper can be solved exactly.} 
\[
\begin{array}{ll}
    \bv_i \cdot \bv_j \leq -\frac{1}{2}, & \forall \{i, j\} \in E,\,  \\
    \|\bv_i\|^2 = 1, & \forall i \in V.
\end{array}
\]
A set of vectors satisfying the equations above is called a \emph{vector 3-coloring} of $G$. The algorithm then uses the following procedure to obtain a large independent set from $G$: 

\begin{figure}[ht!]
	\begin{mdframed}
        {\bf Input: } a graph $G = (V, E)$, a vector 3-coloring $\{\bv_i \mid i \in V\}$, and a parameter $t > 0$.
        
        {\bf Output: } an independent set in $G$.

        \vspace{0.05in}

        $\KMS(G, \{\bv_i\})$
        \begin{enumerate}
            \item Sample a random Gaussian vector $\br$.
            \item Let $S \gets \{i \in V \mid \br \cdot \bv_i \geq t\}$.
            \item Return the set of isolated points in the subgraph induced by $S$.
        \end{enumerate}
	\end{mdframed}
    \caption{The $\KMS$ Algorithm for 3-Colorable Graphs}
\end{figure}
The parameter $t$ is chosen so that $\Pr[\br \cdot \bv_i \geq t]$ is roughly on the order of $\Delta^{-1/3}$ where $\Delta$ is the maximum degree in $G$, so $\E[|S|] \approx n \cdot \Delta^{-1/3}$. A simple calculation shows that with this choice of $t$, for any edge $\{ i, j \}$, we have 
$\Pr[\br \cdot \bv_j \geq t \mid \br \cdot \bv_i \geq t] \approx \frac{1}{\Delta}$,
which means a constant fraction of vertices (in $S$) would in fact be isolated inside $S$. This gives the following theorem.

\begin{theorem}[\cite{karger_approximate_1998}]\label{thm:kms_intro}
    Given an $n$-vertex graph $G = (V, E)$ with maximum degree $\Delta$ and a vector 3-coloring of $G$, the $\KMS$ algorithm produces an independent set of size $\widetilde{\Omega}(n \cdot \Delta^{-1/3})$.
\end{theorem}

By choosing $\Delta \approx n^{3/5}$, Theorems~\ref{thm:KTY} and~\ref{thm:kms_intro} make progress towards an $n^{1/5 + o(1)}$-coloring in the dense and sparse cases, respectively, so their combination gives us an $n^{1/5 + o(1)}$-coloring.

\paragraph{Improvement over KMS from Second-Level Neighborhood} The first improvement over KMS was obtained by Arora,  Chlamtáč, and Charikar~\cite{arora_new_2006} and by Chlamtáč~\cite{chlamtac_approximation_2007}. They made the observation that, in the KMS algorithm, if the returned set is much smaller than $S$ (i.e., most vertices in $S$ have a neighbor in $S$), then we can infer additional geometric structure on the vector 3-coloring $\{\bv_i\}$. This additional structure may then be used to extract a large independent set from the second-level neighborhood of some vertex $i \in V$ (i.e., vertices that are at distance two from $i$).

Let us now discuss this idea in more detail. For simplicity, let us assume that we have a \emph{strict vector 3-coloring}, that is, $\bv_i \cdot \bv_j = -1/2$ for every $\{i, j\} \in E$. This means we can write $\bv_j = -\frac{1}{2}\bv_i + \frac{\sqrt{3}}{2}\bv_{ij}$ for some unit vector $\bv_{ij} \perp \bv_i$. We apply the KMS algorithm with a slightly smaller $t$, so that $\Pr[\br \cdot \bv_i \geq t] \approx \Delta^{-\frac{1}{3(1+c)}}$ for some $c > 0$. This means we include slightly more vertices in $S$ in Step 2, giving $\E[|S|] \approx n \cdot \Delta^{-\frac{1}{3(1+c)}}$. If most of the vertices in $S$ ended up being returned in Step 3, then this is already an improved guarantee over KMS with the original parameter. So we may assume that most vertices are not returned, which means that for a typical $i \in V$, if $i$ is included in $S$, then with high probability some $j \in N(i)$ is also included in $S$ (we use $N(i)$ to denote the neighbors of $i$). 
 Using the decomposition $\bv_j = -\frac{1}{2}\bv_i + \frac{\sqrt{3}}{2}\bv_{ij}$, we can write 
\begin{align}
\Pr[\exists j \in N(i): \br \cdot \bv_j \geq t \mid \br \cdot \bv_i \geq t] & \leq \Pr[\exists j \in N(i): \br \cdot \bv_{ij} \geq \sqrt{3}t \mid \br \cdot \bv_i \geq t]\\ & = \Pr[\exists j \in N(i): \br \cdot \bv_{ij} \geq \sqrt{3}t].
\end{align}
The preceding discussion implies that for a typical vertex $i$, the event ``$\exists j \in N(i): \br \cdot \bv_{ij} \geq \sqrt{3}t$'' happens with high probability. Such a set $\{\bv_{ij} \mid j \in N(i)\}$ is referred to as a \emph{$(\sqrt{3}t, \Omega(1))$-cover}, where $\Omega(1)$ represents the probability that the event happens.
With a straightforward calculation, it can be shown that $\Pr[\br \cdot \bv_{ij} \geq \sqrt{3}t] \approx \Delta^{-\frac{1}{1 + c}}$. By the union bound, the cover size $|N(i)|$ must therefore be on the order of at least $\Delta^{\frac{1}{1 + c}}$ in order to make the event ``$\exists j \in N(i): \br \cdot \bv_{ij} \geq \sqrt{3}t$'' happen with decent probability. Since we have $|N(i)| \leq \Delta$ by the maximum degree assumption, we can only use slightly more vectors in the cover than what is required by the union bound. Such a cover is termed \emph{$c$-inefficient}. Consequently, the vectors in this $c$-inefficient cover must be somewhat ``random''-looking, in the sense that they are spread out in the space and cannot be too tightly clustered in any given direction (for otherwise the probability $ \Pr[\exists j \in N(i): \br \cdot \bv_{ij} \geq \sqrt{3}t]$ would drop significantly).

Now let us consider a two-step walk $i \to j \to k$ in the graph. By the previous argument, in a typical two-step walk where both $\{\bv_{ij} \mid j \in N(i)\}$ and $\{\bv_{jk} \mid k \in N(j)\}$ are $c$-inefficient covers, we would expect $\bv_{ij} \cdot \bv_{jk} \approx 0$ for most choices of $j$ and $k$. This means that in the decomposition
\[
\bv_k = -\frac{1}{2}\bv_j + \frac{\sqrt{3}}{2}\bv_{jk} = \frac{1}{4}\bv_i - \frac{\sqrt{3}}{4}\bv_{ij} + \frac{\sqrt{3}}{2}\bv_{jk},
\]
the vectors $\bv_i$, $\bv_{ij}$, $\bv_{jk}$ are pairwise (roughly) orthogonal. This leads to two consequences.

\begin{itemize}
    \item Using $N^{(2)}(i)$ to denote the second-level neighborhood of $i$, the event ``$\exists k \in N^{(2)}(i): - \frac{\sqrt{3}}{4}\bv_{ij} + \frac{\sqrt{3}}{2}\bv_{jk} \geq \frac{3\sqrt{3}}{4}t$'' happens with decent probability. This is called \emph{cover composition} since we are composing two covers. Indeed, the composition should hold intuitively because both  ``$\exists j \in N(i): \br \cdot (-\bv_{ij}) \geq \sqrt{3}t$'' (note that we replaced $\bv_{ij}$ with $-\bv_{ij}$ using Gaussian rotational symmetry) and ``$\exists k \in N(j): \br \cdot \bv_{jk} \geq \sqrt{3}t$'' happen with decent probability due to the cover property, and they should be somewhat independent. (This is, of course, very informal, as we need $j$ in the second event to come from the outcome of the first event; see Lemma~\ref{lem:cover_composition} for a more formal statement of the cover composition lemma, where, in fact, an $O(\sqrt{c}t)$ term is lost due to technical reasons.) After renormalizing, this shows that the set of vectors $\{\frac{\bv_{k} - (\bv_i \cdot \bv_k)\bv_i}{\|\bv_{k} - (\bv_i \cdot \bv_k)\bv_i\|} \mid k \in N^{(2)}(i)\}$ form a $(\frac{9}{\sqrt{15}}t, \Omega(1)$)-cover. Using the union bound as before, we can deduce that $N^{(2)}(i)$ contains many more vertices than $N(i)$. 
    \item Another thing we obtain is that $\bv_i \cdot \bv_k \approx 1/4$ for a typical $k \in N^{(2)}(i)$. For such typical $k$'s, we can obtain a \emph{vector 2-coloring} (see Lemma~\ref{lem:vector_coloring_positive} for details), which is a very strong property that allows us to extract an independent set of size $\widetilde{\Omega}(|N^{(2)}(i)|)$.
\end{itemize}

Summarizing, this heuristic argument shows that for a typical vertex $i$, (1) the second-level neighborhood $N^{(2)}(i)$ is large, and (2) we can extract a large independent set from $N^{(2)}(i)$ proportional to its size. Making all these heuristic ideas rigorous requires a considerable amount of technical effort. We give a more rigorous overview of the argument in Section~\ref{sec:two_steps}. This improvement, combined with Theorem~\ref{thm:KTY}, yields the current best algorithmic result.

\begin{theorem}[\cite{kawarabayashi_better_2024}]\label{thm:KTY_combined}
    There is a polynomial-time algorithm which, given a $3$-colorable graph $G$, produces a proper coloring with $\widetilde{O}(n^{0.19747})$ colors.
\end{theorem}

\paragraph{Extending to Third-Level Neighborhood}

In this work, we extend the analysis to three-step walks $i \to j \to k \to \ell$ in the graph. The plan of attack is similar, namely, by analyzing these walks, we show that (1) the third-level neighborhood $N^{(3)}(i)$ is large, and (2) we can extract a large independent set from $N^{(3)}(i)$. We now discuss the main technical challenges in these two components.

Let us decompose $\bv_\ell$ similarly as before, and write
\[
\bv_\ell = -\frac{1}{8}\bv_i + \frac{\sqrt{3}}{8}\bv_{ij} - \frac{\sqrt{3}}{4}\bv_{jk} + \frac{\sqrt{3}}{2}\bv_{k\ell}.
\]
The first challenge is to make sure that the vectors $\bv_i, \bv_{ij}, \bv_{jk}$, and $\bv_{k\ell}$ are all pairwise (roughly) orthogonal. This requires a lot more technical effort than the second-level analysis. Notably, ensuring that $\bv_{ij}$ and $\bv_{k\ell}$ are roughly orthogonal requires a nontrivial extension of the corresponding argument from the second level (see Section~\ref{subsec:more_pruning} for more details).

The second challenge comes from composing the two-step cover obtained from the previous second-level neighborhood analysis with an extra step. This is needed to show that $N^{(3)}(i)$ is typically large. To do so, we again apply the cover composition lemma to compose two covers (see Lemma~\ref{lem:cover_composition}). However, as noted earlier, applying this lemma incurs a loss term $O(\sqrt{c'}t)$, where $c'$ is the efficiency of the first cover (recall that efficiency measures how many more vectors are in the cover than are required by the union bound). In the earlier second-level analysis, $c' = c$ is controlled by the threshold $t$ that we choose in the KMS algorithm. In our case, $c'$ is the efficiency of the two-step cover, which a priori doesn't have any guarantee on it. As a result, we may lose a lot in the error term $O(\sqrt{c'}t)$ and not gain any advantage over the second-level analysis.

We get around this issue by using a ``win-win'' argument. Note that if $c'$ is indeed too large for cover composition, then it means that the two-step cover is very inefficient, in the sense that it contains many more vectors than the pessimistic estimate from the union bound. This already yields an improvement over the second-level analysis, as now we have a lot more vertices in $N^{(2)}(i)$, on which we can then apply the same vector 2-coloring argument as before and extract an independent set whose size is meaningfully larger than the previous analysis estimates it to be. On the other hand, if $c'$ is small, then we can apply cover composition without too much loss, which would show that $N^{(3)}(i)$ grows significantly over $N^{(2)}(i)$. This is the main focus of Section~\ref{subsec:win-win}.

The third challenge is that, unlike at the second level, we can no longer define a vector 2-coloring for typical vertices on the third-level neighborhood. This was what enabled us to extract a large independent set from the second-level neighborhood. To overcome this, we design a new vector $5/2$-coloring on $N^{(3)}(i)$, which we highlight next. This new vector $5/2$-coloring would allow us to obtain an independent set of size $\approx |N^{(3)}(i)| \cdot \Delta^{-1/5}$ (using Theorem~\ref{theorem:kms}). This $\Delta^{-1/5}$ loss term is easily compensated by the gain in the size of $N^{(3)}(i)$, thereby completing the argument.

\paragraph{New Vector Coloring on the Third-Level Neighborhood} To design a good vector coloring on $N^{(3)}(i)$, we need to come up with unit vectors $\bu_\ell$ for $\ell\in N^{(3)}(i)$ such that for any $\ell_1, \ell_2 \in N^{(3)}(i)$ that share an edge, $\bu_{\ell_1} \cdot \bu_{\ell_2}$ is as negative as possible (see Definition~\ref{def:vec_coloring} for the definition of a vector coloring). We achieve this using the Sum-of-Squares/Lasserre hierarchy, which is a family of increasingly stronger SDP relaxations, in particular stronger than the one used by the KMS algorithm. One particular feature is that it comes with \emph{local distributions}. Each local distribution is a distribution defined over proper colorings of some constant-sized subset of vertices, and the probabilities in these local distributions can be represented by inner products between the corresponding SDP vectors. As an example, from the hierarchy we may obtain vectors $\bv_{(i, \Colour)}$ for $i \in V$ and color $\Colour$, so that $\bv_{(i, \Colour)} \cdot \bv_{(\ell, \Colour')}$ gives the probability that $i$ is colored $\Colour$ and $\ell$ is colored $\Colour'$. See Section~\ref{subsec:sdp} for further details on this hierarchy.

\begin{figure}[!ht]
\centering
\begin{tikzpicture}[line width=1pt, font=\large]

\def\H{3}
\def\Wleft{1}
\def\Wright{2}

\fill[red!50!white] (0,0) rectangle (\Wleft,\H);
\fill[pattern=dots, pattern color=black!40] (0,0) rectangle (\Wleft,\H);
\draw (0,0) rectangle (\Wleft,\H);
\node[below] at (\Wleft/2,0) {$i$};

\begin{scope}[xshift=3cm]
  \node[below] at (\Wright/4,0) {$\ell_1$};
  \node[below] at (\Wright*3/4,0) {$\ell_2$};

  \fill[red!50!white] (0,0.75*\H) rectangle (\Wright/2,\H);
  \fill[pattern=dots, pattern color=black!40] (0,0.75*\H) rectangle (\Wright/2,\H);

  \fill[blue!50!white] (\Wright/2,0.875*\H) rectangle (\Wright,\H);
  \fill[pattern=crosshatch, pattern color=black!50] (\Wright/2,0.875*\H) rectangle (\Wright,\H);

  \fill[green!50!white] (\Wright/2,0.75*\H) rectangle (\Wright,0.875*\H);
  \fill[pattern=horizontal lines, pattern color=black!40] (\Wright/2,0.75*\H) rectangle (\Wright,0.875*\H);

  \fill[blue!50!white] (0,0.625*\H) rectangle (\Wright/2,0.75*\H);
  \fill[pattern=crosshatch, pattern color=black!50] (0,0.625*\H) rectangle (\Wright/2,0.75*\H);

  \fill[green!50!white] (0,0.5*\H) rectangle (\Wright/2,0.625*\H);
  \fill[pattern=horizontal lines, pattern color=black!40] (0,0.5*\H) rectangle (\Wright/2,0.625*\H);

  \fill[red!50!white]  (\Wright/2,0.5*\H) rectangle (\Wright,0.75*\H);
  \fill[pattern=dots, pattern color=black!40] (\Wright/2,0.5*\H) rectangle (\Wright,0.75*\H);

  \fill[blue!50!white] (0,0.25*\H) rectangle (\Wright/2,0.5*\H);
  \fill[pattern=crosshatch, pattern color=black!50] (0,0.25*\H) rectangle (\Wright/2,0.5*\H);

  \fill[green!50!white]  (\Wright/2,0.25*\H) rectangle (\Wright,0.5*\H);
  \fill[pattern=horizontal lines, pattern color=black!40] (\Wright/2,0.25*\H) rectangle (\Wright,0.5*\H);

  \fill[green!50!white] (0,0) rectangle (\Wright/2,0.25*\H);
  \fill[pattern=horizontal lines, pattern color=black!40] (0,0) rectangle (\Wright/2,0.25*\H);

  \fill[blue!50!white]  (\Wright/2,0) rectangle (\Wright,0.25*\H);
  \fill[pattern=crosshatch, pattern color=black!50] (\Wright/2,0) rectangle (\Wright,0.25*\H);

  \draw (0,0) rectangle (\Wright,\H);
  \draw (\Wright/2,0) -- (\Wright/2,\H);
\end{scope}

\end{tikzpicture}
\caption{Local distribution over an edge $\{\ell_1, \ell_2\}$ conditioned on $i$ being red}\label{fig:local_dist}
\end{figure}

We now aim to construct the vectors $\bu_\ell$ from the SDP vectors $\bv_{(\ell, \Colour)}$, guided by local distributions. For a generic vertex $\ell$ in the third-level neighborhood of $i$, we expect $\bv_i \cdot \bv_{\ell} \approx (-1/2)^3 = -1/8$. Let us assume for simplicity that equality is achieved. In the local distribution over $i$ and $\ell$, the probability that $i$ and $\ell$ have the same color is equal to $\frac{2}{3}\bv_i \cdot \bv_{\ell} + \frac{1}{3} = \frac{1}{4}$. For simplicity, let us assume that $i$ is colored red. Then $\Pr[\ell \text{ is red}] = 1/4$. For $\ell_1, \ell_2 \in N^{(3)}(i)$ that share an edge, their joint local distribution can then be shown in Figure~\ref{fig:local_dist}. Here we have three vertical bars, one for each of the three variables $i, \ell_1, \ell_2$. To sample an outcome, we take a uniformly random horizontal line which intersects these bars, and the color we encounter in a bar is the color for the corresponding variable. Now, we will define $\bu_\ell$ using a linear combination of the vectors $\bv_{(\ell, \Red)}, \bv_{(\ell, \Green)}$, and $\bv_{(\ell, \Blue)}$, namely 
\[
\bu_\ell = \alpha_\Red  \bv_{(\ell, \Red)} + \alpha_\Green \bv_{(\ell, \Green)} + \alpha_\Blue\bv_{(\ell, \Blue)},
\]
where $\alpha_\Red^2\|\bv_{(\ell, \Red)}\|^2 + \alpha_\Green^2\|\bv_{(\ell, \Green)}\|^2 + \alpha_\Blue^2\|\bv_{(\ell, \Blue)}\|^2 = 1$. This gives us
\begin{align}
    \bu_{\ell_1} \cdot \bu_{\ell_2} = \sum_{\Colour \neq \Colour' \in \Colours}\alpha_\Colour \alpha_{\Colour'}\Pr[\ell_1 \text{ is colored } \Colour \text{ and } \ell_2 \text{ is colored } \Colour'].
\end{align}
This is essentially the expected value of $\alpha_\Colour \alpha_{\Colour'}$ if we sample a pair of colors $(\Colour, \Colour')$ according to the local distribution for $\{\ell_1, \ell_2\}$. By inspecting Figure~\ref{fig:local_dist}, it is clear that the best way to minimize this expected value is to have $\alpha_\Green  = -\alpha_\Blue$ and $\alpha_\Red = 0$, which leads to the definition $\bu_{\ell} = \frac{\bv_{(\ell, \Green)} - \bv_{(\ell, \Blue)}}{\|\bv_{(\ell, \Green)} - \bv_{(\ell, \Blue)}\|}$. It can be easily shown that $\|\bv_{(\ell, \Green)} - \bv_{(\ell, \Blue)}\| = \sqrt{\Pr[\ell \text{ is }\Green \text{ or } \Blue]} \approx \sqrt{3/4}$, so $\bu_{\ell_1} \cdot \bu_{\ell_2} = \frac{-\Pr[\ell_1 \text{ is } \Green \text{ and } \ell_2 \text{ is } \Blue]-\Pr[\ell_1 \text{ is } \Blue \text{ and } \ell_2 \text{ is } \Green]}{3/4} = \frac{-1/2}{3/4} = -\frac{2}{3}$, which gives us a vector $5/2$-coloring. 

A more rigorous analysis and a more general discussion of our new vector coloring construction can be found in Section~\ref{sec:vector_colorings}.

\subsection{Other Related Work}
Coloring 3-colorable graphs is a special case of the approximate graph coloring (AGC) problem, where the goal is to color $c$-colorable graphs with $d$ colors, for some $3 \leq c \leq d$. AGC is an important, perhaps the most prominent, example of \emph{promise constraint satisfaction problems} (Promise CSP). Recent advancements in the study of Promise CSPs have led to the strongest known NP-hardness results for AGC. For $c = 3$, Barto, Bulín, Krokhin, and Opršal~\cite{barto2021algebraic} proved NP-hardness for $d = 5$. For $c \geq 4$, Krokhin, Opršal, Wrochna, and Živný~\cite{krokhin2023topology} established NP-hardness for $d = \binom{c}{\lfloor c/2 \rfloor} - 1$. It is also known that some powerful hierarchy-based algorithms cannot solve AGC for any $3 \leq c \leq d$~\cite{ciardo2025approximate}.

For {\em almost coloring} (i.e., one needs to color only $(1-\eps)$ fraction of vertices for constant $\eps > 0$), stronger hardness results are known~\cite{dinur2014derandomized}, most notably the NP-hardness of distinguishing 4-almost-colorable graphs and $c$-almost-colorable graphs for any constant $c$, based on the 2-to-2 Game Theorem~\cite{khot2017independent, dinur2018towards, khot2023pseudorandom}. 

Given the difficulty for worst-case instances, another interesting line of work has focused on algorithms for 3-colorable graphs that satisfy additional structural guarantees, such as those arising from planted and (semi-)random models \cite{BS95,AK97, FK01} or graphs with expansion properties and low threshold ranks \cite{AG11, DF16, KLT17, BHK25, Hsieh26}.

\subsection{Organization of the Paper}
The remainder of the paper is organized as follows. We set up notation, formulate the SDP relaxations and state some background results in Section~\ref{sec:prelim}. We present our new vector coloring constructions in Section~\ref{sec:vector_colorings}. In Section~\ref{sec:two_steps}, we give a more technical overview of the second-level neighborhood analysis by Arora,  Chlamtáč, and Charikar~\cite{arora_new_2006} and by Chlamtáč~\cite{chlamtac_approximation_2007}. 
In Section~\ref{sec:three_steps}, we present our analysis for third-level neighborhoods and prove Theorem~\ref{theorem:main}. We conclude our paper with a few future directions in Section~\ref{sec:conclusion}.

\section{Preliminaries}\label{sec:prelim}

Given a graph $G = (V, E)$ and a vertex $i \in V$, we use $N(i)$ to denote the set of neighbors of $i$ in $G$ (we write $N_G(i)$ to specify the underlying graph when there are multiple graphs over the same vertex set). For a subset of vertices $S \subseteq V$, we define their neighborhood $N(S) = \{j \in V \mid \exists i \in S, j \in N(i)\}$. For $k \geq 1$, we inductively define the $k$th-level neighborhood of a vertex $i \in V$ by taking $N^{(1)}(i) = N(i)$ and $N^{(k+1)}(i) = N(N^{(k)}(i))$. We use $\Delta(G)$ to denote the maximum degree of any vertex in $G$.

Let $\Phi$ and $\varphi$ be the c.d.f. and p.d.f. of the standard Gaussian distribution $\mathcal{N}(0, 1)$. For every $t \in \mathbb{R}$, we use $\tail(t) \coloneqq 1 - \Phi(t) = \Phi(-t)$ to denote the Gaussian tail probability. We use $\br$ to denote a random vector whose entries are independently sampled from $\mathcal{N}(0, 1)$.

\begin{fact}\label{fact:gaussian_tail}
    For every $t > 0$, we have 
    \begin{equation}
        \left(\frac{1}{t} - \frac{1}{t^3}\right) \varphi(t) \leq \tail(t) \leq \frac{1}{t} \varphi(t).
    \end{equation}
\end{fact}

An immediate corollary is that when $t$ is (poly-)logarithmic in $n$, for any constant $c > 0$ we have $\tail(ct) = \widetilde{\Theta} (\tail(t)^{c^2})$.

\subsection{SDP Relaxations and the KMS Algorithm}\label{subsec:sdp}

For any $k \geq 2$, the $k$-colorability problem can be relaxed to the following semi-definite program.
\[
\begin{array}{ll}
    \bv_i \cdot \bv_j \leq -\frac{1}{k-1}, & \forall \{i, j\} \in E,\,  \\
    \|\bv_i\|^2 = 1, & \forall i \in V.
\end{array}
\]
The intended solution is to assign each color class to a vertex of the $(k-1)$-simplex embedded in $\mathbb{R}^{k-1}$. Though the relaxation is written with integer values of $k$ in mind, the equations still make sense if we relax it to fractional numbers. This leads to the following definition.

\begin{definition}\label{def:vec_coloring}
    Let $G = (V, E)$ be a graph. Then, for any real number $\kappa \geq 2$, a set of unit vectors $\{\bv_i \mid i \in V\}$ is called a vector $\kappa$-coloring of $G$ if for every $\{i, j\} \in E$, $\bv_i \cdot \bv_j \leq -\frac{1}{\kappa - 1}$. We say that the coloring is strict if equality is achieved on every edge.
\end{definition}

Given a vector $\kappa$-coloring of a graph, we can use the $\KMS$ algorithm to obtain an independent set, whose size depends on $\kappa$.

\begin{figure}[ht!]
	\begin{mdframed}
        {\bf Input: } a graph $G = (V, E)$, a vector $\kappa$-coloring $\{\bv_i \mid i \in V\}$, and the parameter $\kappa$.
        
        {\bf Output: } an independent set in $G$.

        \vspace{0.05in}

        $\KMS(G, \{\bv_i\}, \kappa)$
        \begin{enumerate}
            \item Sample a random Gaussian vector $\br$.
            \item Let $S \gets \{i \in V \mid \br \cdot \bv_i \geq \sqrt{2\cdot\frac{\kappa - 2}{\kappa} \log (\Delta(G))}\}$.
            \item Return the set of isolated points in the subgraph induced by $S$.
        \end{enumerate}
	\end{mdframed}
    \caption{The $\KMS$ Algorithm}
\end{figure}

\begin{theorem}[\cite{karger_approximate_1998}]\label{theorem:kms}
    For any $\kappa \geq 2$, given a graph $G = (V, E)$ and a vector $\kappa$-coloring of $G$, the $\KMS$ algorithm produces an independent set of size $\widetilde{\Omega}(|V| \cdot \Delta(G)^{-\frac{\kappa - 2}{\kappa}})$.
\end{theorem}
In particular, when $\kappa = 3$, the $\KMS$ algorithm produces an independent set of size $\widetilde{\Omega}(|V| \cdot \Delta(G)^{-\frac{1}{3}})$.

Given some vector coloring $\{\bv_i\}$, we will often want to project some $\bv_j$ along some other $\bv_i$ and an orthogonal direction. We now define the notation for this.

\begin{definition}
     Given a graph $G = (V, E)$ and unit vectors $\{\bv_i \mid i \in V\}$, we write 
     \[
     \bv_{ij} = \frac{\bv_j - (\bv_i \cdot \bv_j)\bv_i}{\|\bv_j - (\bv_i \cdot \bv_j)\bv_i\|},
     \]
     for any $i, j \in V$ such that $\|\bv_j - (\bv_i \cdot \bv_j)\bv_i\| \neq 0$.
\end{definition}
 Note that $\bv_{ij}\perp \bv_i$ and $\bv_j =  (\bv_i \cdot \bv_j)\bv_i + (\bv_{ij} \cdot \bv_{j})\bv_{ij}$.  Also note that in general $\bv_{ij} \neq \bv_{ji}$.
When $\{\bv_i \mid i \in V\}$ is a strict vector 3-coloring, we have that $\bv_{ij} = \frac{2\bv_j + \bv_i}{\sqrt{3}}$, or equivalently $\bv_j = -\frac{1}{2}\bv_i + \frac{\sqrt{3}}{2}\bv_{ij}$, for any $\{i, j\} \in E$.

\begin{fact}\label{fact:vij_vji}
    If $\bv_i \cdot \bv_j = t \in (-1, 1)$, then $\bv_{ij} = \sqrt{1 - t^2}\bv_j - t \bv_{ji}$.
\end{fact}
\begin{proof}
    Since $\bv_i \cdot \bv_j = t$, we have $\bv_j = t\bv_i + \sqrt{1 - t^2}\bv_{ij}$ and $\bv_i = t\bv_j + \sqrt{1 - t^2}\bv_{ji}$. It follows that 
\[
        \sqrt{1 - t^2} \bv_{ij} = \bv_j - t \bv_i = \bv_j - t (t\bv_j + \sqrt{1 - t^2}\bv_{ji}) = (1 - t^2) \bv_j - t\sqrt{1 - t^2} \bv_{ji}.
  \]
    The fact follows by dividing both sides by $\sqrt{1 - t^2}$.
\end{proof}

To make progress over the KMS algorithm, we will consider a stronger SDP relaxation given by the Sum-of-Squares(SoS)/Lasserre hierarchy~\cite{lasserre2001global}. To formally state it, let us first formulate the 3-colorability problem using the following 0-1 program. Given a graph $G = (V, E)$, we write:
\[
\begin{array}{ll}
     X_{(i, \Red)} + X_{(i, \Green)} + X_{(i, \Blue)} = 1, & \forall i \in V,\\
    X_{(i, \Colour)} X_{(j, \Colour)} = 0, & \forall \{i, j\} \in E,\,  \Colour \in \{\Red, \Green, \Blue\}. \\
    X_{(i,  \Colour)} \in \{0, 1\}, & \forall i \in V, \, \Colour \in \{\Red, \Green, \Blue\}. 
\end{array}
\]
Here the variable $X_{(i, \Colour)}$ serves as the indicator of ``$i$ takes color $\Colour$''. The first constraint enforces that each variable $i \in V$ picks exactly one color, and the second constraint makes sure that the two endpoints of an edge do not take the same color. We now relax it to the following semi-definite program. Let $k > 0$ be an integer. For any $S \subseteq V \times \{\Red, \Green, \Blue\}$ with $|S| \leq k$, we have a vector-valued variable $\bv_S$. These vectors satisfy the following equations:
\[
\begin{array}{ll}
    \|\bv_\varnothing\| = 1, & \\
    \bv_{(i, \Red)} + \bv_{(i, \Green)} + \bv_{(i, \Blue)} = \bv_{\varnothing}, &  \forall i \in V, \\
    \bv_{(i, \Colour)} \cdot \bv_{(i, \Colour')} = 0, & \forall i \in V,\,  \Colour, \Colour' \in \Colours, \Colour \neq \Colour',\\
    \bv_{(i, \Colour)} \cdot \bv_{(j, \Colour)} = 0, & \forall \{i, j\} \in E,\,  \Colour \in \Colours, \\
    \bv_{S_1} \cdot \bv_{S_2} = \bv_{S_3} \cdot \bv_{S_4}, & \forall S_1, S_2, S_3, S_4 \subseteq V \times \Colours, S_1 \cup S_2 = S_3 \cup S_4, |S_1 \cup S_2| \leq k.
\end{array}
\]
This is called the \emph{$k$-round Sum-of-Squares(SoS)/Lasserre relaxation}. As long as $k$ is a constant, this SDP has polynomial size, so it can be solved in polynomial time. One very helpful intuitive interpretation of these vectors is via \emph{local distributions}. 
The SDP computes, for each $S \subseteq V \times \{\Red, \Green, \Blue\}$ with $|S| \leq k$, a distribution over $\{X_{(i, \Colour)} \mid (i, \Colour) \in S\}$, which we call a local distribution. These local distributions satisfy consistency requirements, namely for any $S_1$ and $S_2$ with $S_1 \cap S_2 \neq \varnothing$, the marginals of their local distributions on $S_1 \cap S_2$ should be the same. It is sometimes useful to imagine that these local distributions were marginals of some global distribution over the entire set $\{X_{(i, \Colour)} \mid (i, \Colour) \in V \times  \{\Red, \Green, \Blue\}\}$, though we should keep in mind that this global distribution may not necessarily exist. 

The interpretation of these inner products using the local distributions is that for any $S_1, S_2$,
\begin{equation}
    \bv_{S_1} \cdot \bv_{S_2} = \Pr[\forall (i, \Colour) \in S_1 \cup S_2, i = \Colour],
\end{equation}
where the probability on the right-hand side is over the local distribution and we use $i = \Colour$ to mean $X_{(i, \Colour)} = 1$ (i.e., $i$ takes color $\Colour$), slightly abusing the notation.

Observe that due to the symmetry between the three colors, we may further impose that the inner products between these vectors are preserved up to permutation of the colors. In particular, this implies that for every $i \in V$ and $\Colour \in \{\Red, \Green, \Blue\}$, $\|\bv_{(i, \Colour)}\|^2 = \Pr[i = \Colour] = \frac{1}{3}$. This will be assumed in the remainder of this paper.

It is straightforward to obtain a vector $3$-coloring from the stronger SoS/Lasserre relaxation. Indeed, for every $i \in V$, we have $\bv_{(i, \Red)} \cdot \bv_{\varnothing} =  \Pr[i = \Red] = \frac{1}{3}$, so we can write
\begin{equation}\label{eq:vector_3_coloring_def}
    \bv_{(i, \Red)} = \frac{1}{3}\bv_\varnothing + \frac{\sqrt{2}}{3} \bv_i
\end{equation}
for some unit vector $\bv_i \perp \bv_\varnothing$. For every $\{i, j\} \in E$, we have 
\begin{equation}
    0 = \Pr[i = \Red, j = \Red] = \bv_{(i, \Red)} \cdot \bv_{(j, \Red)} = \frac{1}{9} + \frac{2}{9}\bv_i \cdot \bv_j,
\end{equation}
so $\bv_i \cdot \bv_j = -\frac{1}{2}$. It follows that the set of vectors $\{\bv_i\}$ forms a strict vector 3-coloring. We will assume in the following sections that our vector 3-colorings are obtained this way.

\subsection{Combining Combinatorial and SDP-Based Algorithms}\label{subsec:combine}

To formally combine combinatorial and SDP-based algorithms, we need to first describe Blum's notion of progress towards an $f(n)$-coloring (of 3-colorable graphs)~\cite{blum_new_1994}, which consists of the following three types:
\begin{itemize}
    \item Type 1 {\sf [Large-IS]}: find an independent set of size $\Omega(n / f(n))$.
    \item Type 2 {\sf [Small-Nbhd]}: find an independent set $S$ such that $|N(S)| = O(f(n)|S|)$.
    \item Type 3 {\sf [Same-Color]}: find two distinct vertices $i, j$ which must be the same color in any proper 3-coloring of the graph.
\end{itemize}

The main idea is that if, for any 3-colorable graph, we can make one of these types of progress in polynomial time, then we can compute an $f(n)$-coloring. Blum and Karger~\cite{blum_tildeon314-coloring_1997} observed that we may make progress separately for graphs with large or small degrees. In this paper, we will use the following theorem by Kawarabayashi and Thorup as a black-box tool to perform this combination.

\begin{theorem}[cf. Proposition 7.1 in~\cite{kawarabayashi_coloring_2017}]\label{thm:KT_combination}
    Let $\alpha, \beta \in (0, 1)$. Suppose that in polynomial time we can make progress of any of the three types toward $\Tilde{O}(n^{\alpha})$ on an $n$-vertex graph $G$ as long as either the maximum degree in $G$ is at most $n^\beta$, or the minimum degree in $G$ is at least $n^\beta$. Then we can find a $\widetilde{O}(n^\alpha)$-coloring in polynomial time.
\end{theorem}

\begin{corollary}\label{cor:sdp_combination}
    Let $c > 0$ be some constant. Suppose that we can always make progress toward an $\widetilde{O}(n^\frac{1}{5+3c})$-coloring on an $n$-vertex graph with $\Delta(G) \leq n^{\frac{3+3c}{5+3c}}$, then, for any $0 < c' < c$ we can find an $O(n^\frac{1}{5+3c'})$-coloring in polynomial time.
\end{corollary}
\begin{proof}
    Let $\beta = \frac{3 + 3c}{5+3c}$, then by Theorem~\ref{thm:KTY} we can make progress towards a $k$-coloring on a graph $G$ with minimum degree $n^{\beta}$ for some $k = n^{o(1) + \frac{1}{2}(1 - \beta)} = n^{\frac{1}{5+3c} + o(1)}$. By our assumption, we can make progress towards a $\widetilde{O}(n^\frac{1}{5+3c})$-coloring on a graph $G$ with maximum degree $n^{\beta}$. So by Theorem~\ref{thm:KT_combination} we can obtain an $O(n^\frac{1}{5+3c'})$-coloring in polynomial time for any $0 < c' < c$ (the tilde can be dropped by slightly adjusting $c'$).
\end{proof}

\section{New Constructions of Vector Colorings}\label{sec:vector_colorings}

In this section, we present our new constructions of vector colorings.
We will frequently use the following simple lemma (without explicitly referring to it) in the proofs.
\begin{lemma}
    Let $G = (V, E)$, $k > 0$, and $\{\bv_S \mid S \subseteq V \times \Colours, |S| \leq k\}$ be a solution to the $k$-round SoS/Lasserre relaxation. Then for any $i \in V$ and $S \subseteq V \times \Colours$ such that $|S| \leq k - 1$, we have
    \begin{equation}
        \|\bv_S\|^2 =  \|\bv_{S \cup \{(i, \Red)\}}\|^2 + \|\bv_{S \cup \{(i, \Green)\}}\|^2 + \|\bv_{S \cup \{(i, \Blue)\}}\|^2.
    \end{equation}
\end{lemma}
\begin{proof}
    From the SoS constraints we have
    \[
         \|\bv_S\|^2 = \bv_S \cdot \bv_\varnothing =  \bv_S \cdot \left( \bv_{(i, \Red)} + \bv_{(i, \Green)} + \bv_{(i, \Blue)}\right) =  \|\bv_{S \cup \{(i, \Red)\}}\|^2 + \|\bv_{S \cup \{(i, \Green)\}}\|^2 + \|\bv_{S \cup \{(i, \Blue)\}}\|^2. \qedhere
    \]
\end{proof}

\begin{lemma}\label{lem:vector_coloring_negative}
    Let $G = (V, E)$ be a graph and assume $\{\bv_i \mid i \in V\}$ is a strict vector 3-coloring obtained using~\eqref{eq:vector_3_coloring_def} from a 3-round SoS/Lasserre relaxation for $G$. 
    Assume that for some $i \in V$ and $t \leq 0$, we have $\bv_i \cdot \bv_j \leq t$ for every $j \in V \setminus \{i\}$. 
    Then, in polynomial time we can find a vector $(\frac{3-6t}{1-4t})$-coloring on the subgraph induced by $V \setminus \{i\}$.
\end{lemma}
\begin{proof}
    By \eqref{eq:vector_3_coloring_def}, for every $j$ such that $\bv_i \cdot \bv_j \leq t$ we have $\frac{9}{2} \bv_{(i, \Red)} \cdot \bv_{(j, \Red)} - \frac{1}{2} \leq t$, or equivalently 
    \begin{equation}
        \|\bv_{(i, \Red), (j, \Red)}\|^2 \leq \frac{1 + 2t}{9}.
    \end{equation}
    Let us now define a new set of vectors $\{\bu_j \mid j \in V \setminus \{i\}\}$ where 
    \begin{equation}
        \bu_j \coloneqq \frac{\bv_{(i, \Red),(j, \Green)} - \bv_{(i, \Red),(j, \Blue)}}{\|\bv_{(i, \Red),(j, \Green)} - \bv_{(i, \Red),(j, \Blue)}\|} = \frac{\bv_{(i, \Red),(j, \Green)} - \bv_{(i, \Red),(j, \Blue)}}{\sqrt{\|\bv_{(i, \Red),(j, \Green)}\|^2 + \|\bv_{(i, \Red),(j, \Blue)}\|^2}}.
    \end{equation}
    The equality holds since $\bv_{(i, \Red),(j, \Green)} \cdot \bv_{(i, \Red),(j, \Blue)} = 0$.

    For every $j, \ell \in V \setminus \{i\}$ that share an edge, we have $\bv_{(i, \Red),(j, \Green)}\cdot\bv_{(i, \Red),(\ell, \Green)} = \bv_{(i, \Red),(j, \Blue)}\cdot\bv_{(i, \Red),(\ell, \Blue)} = 0$ and therefore
    \begin{align}
        \bu_j \cdot \bu_\ell & = \frac{\bv_{(i, \Red),(j, \Green)} - \bv_{(i, \Red),(j, \Blue)}}{\sqrt{\|\bv_{(i, \Red),(j, \Green)}\|^2 + \|\bv_{(i, \Red),(j, \Blue)}\|^2}} \cdot \frac{\bv_{(i, \Red),(\ell, \Green)} - \bv_{(i, \Red),(\ell, \Blue)}}{\sqrt{\|\bv_{(i, \Red),(\ell, \Green)}\|^2 + \|\bv_{(i, \Red),(\ell, \Blue)}\|^2}}\\
        & = - \frac{\bv_{(i, \Red),(j, \Green)} \cdot \bv_{(i, \Red),(\ell, \Blue)}+ \bv_{(i, \Red),(j, \Blue)} \cdot \bv_{(i, \Red),(\ell, \Green)}}{\sqrt{\|\bv_{(i, \Red),(j, \Green)}\|^2 + \|\bv_{(i, \Red),(j, \Blue)}\|^2}\sqrt{\|\bv_{(i, \Red),(\ell, \Green)}\|^2 + \|\bv_{(i, \Red),(\ell, \Blue)}\|^2}} \label{eq:vector_coloring_intermediate_result}
    \end{align}
    For the denominator in \eqref{eq:vector_coloring_intermediate_result}, we have
    \begin{equation}
        \|\bv_{(i, \Red),(j, \Green)}\|^2 + \|\bv_{(i, \Red),(j, \Blue)}\|^2 = \|\bv_{(i, \Red)}\|^2 - \|\bv_{(i, \Red), (j, \Red)}\|^2 = \frac{1}{3} - \|\bv_{(i, \Red), (j, \Red)}\|^2.
    \end{equation}
    For the numerator in \eqref{eq:vector_coloring_intermediate_result}, we first write
    \begin{align}
        &\,\frac{1}{3} - \|\bv_{(i, \Red), (j, \Red)}\|^2 \\
        =\, & \,  \|\bv_{(i, \Red),(j, \Green)}\|^2 +  \|\bv_{(i, \Red),(j, \Blue)}\|^2 \\
        =\, & \,  \|\bv_{(i, \Red),(j, \Green), (\ell, \Red)}\|^2 + \|\bv_{(i, \Red),(j, \Green), (\ell, \Blue)}\|^2 +  \|\bv_{(i, \Red),(j, \Blue), (\ell, \Red)}\|^2 + \|\bv_{(i, \Red),(j, \Blue), (\ell, \Green)}\|^2 \\
        =\, & \,  \|\bv_{(i, \Red), (\ell, \Red)}\|^2 + \|\bv_{(i, \Red),(j, \Green), (\ell, \Blue)}\|^2 +   \|\bv_{(i, \Red),(j, \Blue), (\ell, \Green)}\|^2, \label{eq:RGB_RBG}
    \end{align}
    so the numerator can be written as
    \begin{align}
        \bv_{(i, \Red),(j, \Green)} \cdot \bv_{(i, \Red),(\ell, \Blue)}+ \bv_{(i, \Red),(j, \Blue)} \cdot \bv_{(i, \Red),(\ell, \Green)} & \,=\, \|\bv_{(i, \Red),(j, \Green),(\ell, \Blue)}\|^2 +  \|\bv_{(i, \Red),(j, \Blue),(\ell, \Green)}\|^2  \\& \,=\, \frac{1}{3} - \|\bv_{(i, \Red),(j, \Red)}\|^2 - \|\bv_{(i, \Red), (\ell, \Red)}\|^2.
    \end{align}
    Putting these terms together we obtain 
    \begin{equation}
        \bu_j \cdot \bu_\ell = -\frac{\frac{1}{3} - \|\bv_{(i, \Red), (j, \Red)}\|^2 - \|\bv_{(i, \Red), (\ell, \Red)}\|^2}{\sqrt{\frac{1}{3} - \|\bv_{(i, \Red), (j, \Red)}\|^2}\sqrt{\frac{1}{3} - \|\bv_{(i, \Red), (\ell, \Red)}\|^2}}.
    \end{equation}
    The above expression is maximized when $\|\bv_{(i, \Red), (j, \Red)}\|^2 = \|\bv_{(i, \Red), (\ell, \Red)}\|^2 = \frac{1 + 2t}{9}$, which gives $\bu_j \cdot \bu_\ell = -\frac{1-4t}{2 - 2t}$. This gives a vector $(\frac{3-6t}{1-4t})$-coloring on the subgraph induced by $V \setminus \{i\}$.
\end{proof}
We note that similar constructions are also possible starting from combinatorial assumptions. We state and prove the case of vector $\frac{5}{2}$-colorings as an example, which is straightforward to generalize to other parameters.

\begin{lemma}\label{lem:combinatorial_vector_coloring}
    Let $G = (V, E)$ be an $n$-vertex graph, and let $\col: V \to \Colours$ be a proper 3-coloring of $G$ such that $|\col^{-1}(\Red)| \leq n / 4$. Then, for any $\epsilon = \epsilon(n) \in (0, 1/4)$, we can find the following in polynomial time: \begin{itemize}
        \item a subset $A \subseteq V(G)$ with $|A| \geq c\cdot\epsilon n$ where $c > 0$ is some constant independent of $\epsilon$ and $n$, and
        \item a vector $(\frac{5}{2} + \epsilon)$-coloring of the subgraph induced by $A$.
    \end{itemize}  
\end{lemma}
\begin{proof}
    By our assumption, we may add the constraint $\sum_{i \in V(G)}\|\bv_{(i, R)}\|^2 \leq |V(G)| / 4$ to the SoS hierarchy (while removing the color symmetry constraints) and this is still a feasible relaxation. For some $\epsilon > 0$, consider the set of vertices
    \begin{equation}
        A \coloneqq \left\{i \in V(G) \mid \|\bv_{(i, R)}\|^2 \leq \frac{1}{4} + \epsilon\right\}.
    \end{equation}
    We have
    \begin{equation}
        \frac{n}{4} \geq \sum_{i \in V(G)}\|\bv_{(i, R)}\|^2 \geq (n - |A|) \cdot \left(\frac{1}{4} + \epsilon\right)
    \end{equation}
    from which it follows that $|A| \geq \frac{4\epsilon}{1 + 4\epsilon} n \geq 2\epsilon n$.
    For every $i \in A$, let us define 
    \begin{equation}
        \bu_i \coloneqq \frac{\bv_{(i, G)} - \bv_{(i, B)}}{\|\bv_{(i, G)} - \bv_{(i, B)}\|} = \frac{\bv_{(i, G)} - \bv_{(i, B)}}{\sqrt{\|\bv_{(i, G)}\|^2 + \|\bv_{(i, B)}\|^2}}.
    \end{equation}
    Then for every $i, j \in A$ that share an edge, using calculations similar to Lemma~\ref{lem:vector_coloring_negative} we obtain
    \begin{equation}
        \bu_i \cdot \bu_j = - \frac{\bv_{(i, G)} \cdot \bv_{(j, B)} + \bv_{(i, B)} \cdot \bv_{(j, G)}}{\sqrt{\|\bv_{(i, G)}\|^2 + \|\bv_{(i, B)}\|^2}\sqrt{\|\bv_{(j, G)}\|^2 + \|\bv_{(j, B)}\|^2}} = -\frac{1 - \|\bv_{(i, R)}\|^2 - \|\bv_{(j, R)}\|^2}{\sqrt{1 - \|\bv_{(i, R)}\|^2}\sqrt{1 - \|\bv_{(j, R)}\|^2}}.
    \end{equation}
    The above expression is maximized when $\|\bv_{(i, R)}\|^2 = \|\bv_{(j, R)}\|^2 = \frac{1}{4} + \epsilon$, which gives $\bu_i \cdot \bu_j = -\frac{2 - 8\epsilon}{3 - 4\epsilon} \leq -\frac{2}{3} + \frac{8}{3}\epsilon$. This gives a vector $(\frac{5}{2} + \epsilon)$-coloring of the subgraph induced by $A$ by appropriately rescaling $\epsilon$.
\end{proof}
We can also obtain a $(3 - \Omega(1))$-coloring from vectors that are sufficiently positively correlated with some $\bv_i$. 
\begin{lemma}\label{lem:vector_coloring_positive}
    Let $G = (V, E)$ be a graph and assume $\{\bv_i \mid i \in V\}$ is a strict vector 3-coloring obtained using~\eqref{eq:vector_3_coloring_def} from a 3-round SoS/Lasserre relaxation for $G$. 
    Assume that for some $i \in V$ and $\frac{1}{16} \leq t \leq \frac{1}{4}$, we have $\bv_i \cdot \bv_j \geq t$ for every $j \in V \setminus \{i\}$. 
    Then, in polynomial time we can find a vector $(\frac{4+8t}{1+8t})$-coloring on the subgraph induced by $V \setminus \{i\}$. In particular, if $t = 1/4$, then we can obtain a vector 2-coloring.
\end{lemma}
\begin{proof}
    Similar to Lemma~\ref{lem:vector_coloring_negative}, we obtain from our assumption that 
    \begin{equation}
        \|\bv_{(i, \Red), (j, \Red)}\|^2 \geq \frac{1 + 2t}{9}.
    \end{equation}
    For every $j \in V \setminus \{i\}$, let us define 
    \begin{equation}
        \bu_j \coloneqq \frac{\bv_{(i, \Red),(j, \Red)} - \bv_{(i, \Red),(j, \Green)} - \bv_{(i, \Red),(j, \Blue)}}{\|\bv_{(i, \Red),(j, \Red)} - \bv_{(i, \Red),(j, \Green)} - \bv_{(i, \Red),(j, \Blue)}\|} = \frac{\bv_{(i, \Red),(j, \Red)} - \bv_{(i, \Red),(j, \Green)} - \bv_{(i, \Red),(j, \Blue)}}{1/\sqrt{3}}.
    \end{equation}
    Then, for every neighboring $j, \ell  \in V \setminus \{i\}$ we have 
    \begin{align}
        \bu_j \cdot \bu_\ell &= \frac{\bv_{(i, \Red),(j, \Red)} - \bv_{(i, \Red),(j, \Green)} - \bv_{(i, \Red),(j, \Blue)}}{1/\sqrt{3}} \cdot \frac{\bv_{(i, \Red),(\ell, \Red)} - \bv_{(i, \Red),(\ell, \Green)} - \bv_{(i, \Red),(\ell, \Blue)}}{1/\sqrt{3}} \\
        &= 3 \Big(\|\bv_{(i, \Red), (j, \Green), (\ell, \Blue)}\|^2 + \|\bv_{(i, \Red), (j, \Blue), (\ell, \Green)}\|^2 -\|\bv_{(i, \Red), (j, \Red), (\ell, \Green)}\|^2 \\
        & \qquad \qquad -\|\bv_{(i, \Red), (j, \Red), (\ell, \Blue)}\|^2   -\|\bv_{(i, \Red), (j, \Green), (\ell, \Red)}\|^2  -\|\bv_{(i, \Red), (j, \Blue), (\ell, \Red)}\|^2 \Big)\\ 
        &=3 \Big(\|\bv_{(i, \Red), (j, \Green), (\ell, \Blue)}\|^2 + \|\bv_{(i, \Red), (j, \Blue), (\ell, \Green)}\|^2 -\|\bv_{(i, \Red), (j, \Red)}\|^2   -\|\bv_{(i, \Red), (\ell, \Red)}\|^2 \Big).
    \end{align}
    By \eqref{eq:RGB_RBG}, we have 
    \[
    \|\bv_{(i, \Red), (j, \Green), (\ell, \Blue)}\|^2 + \|\bv_{(i, \Red), (j, \Blue), (\ell, \Green)}\|^2 = \frac{1}{3} - \|\bv_{(i, \Red), (j, \Red)}\|^2   -\|\bv_{(i, \Red), (\ell, \Red)}\|^2.
    \]
    So plugging this in above, we get
    \begin{equation}
        \bu_j \cdot \bu_\ell = 3 \cdot\left( \frac{1}{3} - 2 \left(\|\bv_{(i, \Red), (j, \Red)}\|^2   +\|\bv_{(i, \Red), (\ell, \Red)}\|^2\right) \right) \leq 1 - 12 \cdot \frac{1 + 2t}{9} = -\frac{1+8t}{3}.
    \end{equation}
    So we obtain a vector $(\frac{4+8t}{1+8t})$-coloring on the subgraph induced by $V \setminus \{i\}$. 
\end{proof}

We remark that a more general way to obtain improved vector colorings conditioned on some vertex $i$ is to consider, for any vertex $j$, 
\begin{equation}
    \bu_j = \frac{\alpha_\Red\cdot\bv_{(i, \Red),(j, \Red)} +\alpha_\Green \cdot\bv_{(i, \Red),(j, \Green)} + \alpha_\Blue \cdot\bv_{(i, \Red),(j, \Blue)}}{\|\alpha_\Red\cdot\bv_{(i, \Red),(j, \Red)} +\alpha_\Green\cdot \bv_{(i, \Red),(j, \Green)} + \alpha_\Blue\cdot \bv_{(i, \Red),(j, \Blue)}\|},
\end{equation}
where $\alpha_\Red, \alpha_\Green, \alpha_\Blue$ are some arbitrary coefficients to be optimized. Indeed, this generalizes both Lemma~\ref{lem:vector_coloring_negative} and Lemma~\ref{lem:vector_coloring_positive}. Lemma~\ref{lem:vector_coloring_negative}, which our analysis relies on, does not benefit from this more general formalization. We made no attempt to optimize Lemma~\ref{lem:vector_coloring_positive}. Some special cases of this formalization can also be found in~\cite{arora_new_2006}, where vector $\frac{3}{1 + 2t^2}$-colorings are obtained starting from the same assumption as in Lemma~\ref{lem:vector_coloring_positive} (but with $t$ relaxed to $t > 0$), and in~\cite{chlamtac_approximation_2007}, where a vector 2-coloring is obtained for $t = 1/4$. 

\section{Analyzing Second-Level Neighborhoods}\label{sec:two_steps}

In this section, we give a more detailed overview of the analysis of second-level neighborhoods by Arora,  Chlamtáč, and Charikar~\cite{arora_new_2006} and by Chlamtáč~\cite{chlamtac_approximation_2007}. The main purpose is to highlight the key ideas that will serve as the basis for our refinements in the next section. Most of the definition and theorem statements are taken or adapted from Chlamtáč's PhD thesis~\cite{chlamtac_non-local_2009}.

We already gave a conceptual overview in Section~\ref{sec:overview_intro}. The actual analysis here differs slightly in that we run a variant of the KMS algorithm, which is called $\KMSp$ in~\cite{arora_new_2006} (see Figure~\ref{fig:kmsp}). The only difference is that in Step 3, instead of returning all isolated points in the subgraph induced by $S$, we find and remove vertices in a maximal matching and return the remaining vertices. It is clear that this still gives an independent set.

\begin{figure}[ht!]
	\begin{mdframed}
        {\bf Input: } a graph $G = (V, E)$, a vector 3-coloring $\{\bv_i \mid i \in V\}$ given by~\eqref{eq:vector_3_coloring_def}, and a parameter $t > 0$.
        
        {\bf Output: } an independent set $S \subseteq V$.

        \vspace{0.05in}

        $\KMSp(G, \{\bv_i\}, t)$
        \begin{enumerate}
            \item Sample a random Gaussian vector $\br$.
            \item Let $S \gets \{i \in V \mid \br \cdot \bv_i \geq t\}$.
            \item Find a maximal matching $M$ in $S$, and remove every vertex in $M$ from $S$.
            \item Return $S$.
        \end{enumerate}
	\end{mdframed}
    \caption{The $\KMSp$ Algorithm}\label{fig:kmsp}
\end{figure}

We now formally state Chlamtáč's theorem.

\begin{definition}\label{def:kmsp_failure}
    We say that $\KMSp$ fails on $G$ with parameter $t$ if \[\Pr[\,i \text{ is removed from } S \mid \br \cdot \bv_i \geq t] \geq 1/2\] for at least half of the vertices $i$ in $G$.\footnote{This is clearly a very conservative condition, in that $\KMSp$ could still return an independent set of $\widetilde{\Omega}(n\tail(t))$ when this holds, but this is sufficient for the analysis.} For $c > 0$, we say that the parameter $t$ is $c$-inefficient (for $G$) if $\tail(t) = \Delta(G)^{-\frac{1}{3(1+c)}}$.
\end{definition}

\begin{theorem}[cf. Theorem 3.2.2 in~\cite{chlamtac_non-local_2009}]\label{thm:chlamtac_clean_statement}
    Fix $c > 0$, and let $n$ be sufficiently large. Assume that $\KMSp$ fails on some $n$-vertex graph $G$ with $c$-inefficient parameter $t$. Then, for any $\eta \in (0, \eta_0)$ we can find in polynomial time an independent set in $G$ of size $\Omega\left(\Delta(G)^{\frac{\eta^2}{3(1+c)}}\right)$, where
    \begin{equation}
        \eta_0 = \eta_0(c) = \inf_{0 \leq \alpha \leq \frac{c}{1+c}}\frac{9-3\alpha-6\sqrt{(1 - \alpha^2)c}}{4\sqrt{1 - \left(\frac{1}{4} + \frac{3}{4}\alpha\right)^2}}.
\end{equation}
\end{theorem}
The expression $\Omega(\Delta(G)^{\frac{\eta^2}{3(1+c)}})$ essentially comes from $\tail(\eta t)^{-1}$ via Fact~\ref{fact:gaussian_tail}. Note that when $c = 0$, we have $\eta_0 = \frac{9}{\sqrt{15}}$ which recovers the heuristic argument in Section~\ref{sec:overview_intro}. The parameter $\alpha$ can be thought of as representing the loss hidden in the heuristic argument. We remark that in~\cite{arora_new_2006}, a similar statement was stated with $\alpha \in [-\sqrt{\frac{c}{1+c}}, \sqrt{\frac{c}{1+c}}]$. Much technical work was devoted in~\cite{chlamtac_non-local_2009} in order to shrink the range to  $[0, \frac{c}{1+c}]$, using a process called \emph{pruning}. This improvement will also be crucial for the third-level neighborhood analysis in the next section. In the remainder of this section, we give a more technical overview of the proof of Theorem~\ref{thm:chlamtac_clean_statement}, stating the key concepts and lemmas in preparation for our third-level neighborhood analysis.

We remark that Section 4 of~\cite{chlamtac_non-local_2009} gives a different approach to bound the size of the second-level neighborhood, resulting in a different theorem statement. However, that approach seems to give a slightly worse result in the parameter regime that we are interested in.

\subsection{Covers and Packings}

\begin{definition}
    A set of unit vectors $X$ is called an $(s, \delta)$-cover if 
    \begin{equation}
        \Pr_{\br}[\exists \bx \in X: \br \cdot \bx \geq s] \geq \delta. 
    \end{equation}
    $X$ is said to be (at most) $c$-inefficient with respect to $s$ if $|X| \leq \tail(s)^{-(1 + c)}$. 
\end{definition}

Geometrically, each $\bx \in X$ defines a halfspace $\{\br \mid \br \cdot \bx \geq s\}$ with Gaussian measure $\tail(s)$. The cover condition says that the union of these halfspaces ``covers'' Gaussian measure at least $\delta$.

We will sometimes say that a set $X$ is $c$-inefficient without explicitly referring to the parameter $s$, which should be clear from context.

\begin{fact}\label{fact:cover_lower_bd}
    If $X$ is an $(s, \delta)$-cover, then we have  $|X| \geq \delta \cdot \tail(s)^{-1}$.
\end{fact}
\begin{proof}
    By the union bound and the definition of an $(s, \delta)$-cover we have 
    \[
        \delta \leq \Pr_{\br}[\exists \bx \in X: \br \cdot \bx \geq s] \leq |X| \cdot \tail(s). \qedhere
\]
\end{proof}
A closely related concept which will be frequently used in the analysis is the notion of an $(s,\delta)$-packing. Packings allow us to get a better handle on contributions from different subsets of a cover to the cover probability.
\begin{definition}
    Let $X$ be a set of unit vectors and $\mu$ a measure on $X$. We say that $(X, \mu)$ is an $(s, \delta)$-packing if the following holds:
    \begin{itemize}
        \item For every $X' \subseteq X$, $\mu(X') \leq \Pr[\exists \bx \in X': \br \cdot \bx \geq s]$.
        \item $\mu(X) \geq \delta$.
    \end{itemize}
\end{definition}

Intuitively, $\mu$ packs total mass at least $\delta$ onto the vectors in $X$. The first condition prevents this mass from being too concentrated: every subset can carry at most the Gaussian measure covered by its associated halfspaces.

When $X = \{\bv_i \mid i \in I\}$ is indexed by some index set $I$, we might sometimes slightly abuse the notation and write $\mu(\bv_i)$ or $\mu(i)$ as a shorthand for $\mu(\{\bv_i\})$.
We remark that $(s, \delta)$-packing is an equivalent characterization of $(s, \delta)$-cover, as the following lemma shows.

\begin{lemma}
    $X$ is an $(s, \delta)$-cover if and only if there exists $\mu$ such that $(X, \mu)$ is an $(s, \delta)$-packing.
\end{lemma}
\begin{proof}
    The $\Leftarrow$ direction is straightforward. For the $\Rightarrow$ direction, for every $\bx \in X$, let 
    \begin{equation}
        \mu(\bx) = \Pr_\br[\br \cdot \bx \geq s \text{ and } \br\cdot\bx = \max_{\bx' \in X}\br \cdot \bx'].
    \end{equation}
    Then $(X, \mu)$ is an $(s, \delta)$-packing.
\end{proof}

\begin{example}\label{example:mu_i_p_i}
    Given a graph $G = (V, E)$, by considering the execution of the $\KMSp$ algorithm on it, we can define a packing $(\{\bv_{ij} \mid j \in N(i)\}, \mu_i)$ for every vertex $i \in V$ by letting
    \begin{equation}
        \mu_i(j) = \Pr[\{i, j\} \text{ included in } M \mid \br \cdot \bv_i \geq t],
    \end{equation}
    where $M$ is the maximal matching chosen in step 3 of $\KMSp$. This gives a $(\sqrt{3}t, p_i)$-packing, where $p_i$ is the probability that $i$ is removed in step 3 conditioned on $\br \cdot \bv_i \geq t$ (i.e., it's included in step 2). If $\KMSp$ fails on $G$ (as defined in Definition~\ref{def:kmsp_failure}), then $p_i \geq \frac{1}{2}$ for at least half of the vertices in $G$.
    
    One notable property that this packing enjoys is symmetry, namely, for every $\{i, j\} \in E$, we have $\mu_{i}(j) = \mu_j(i)$. Indeed, we have
    \begin{equation}
        \mu_{i}(j) =  \Pr[\{i, j\} \text{ included in } M \mid \br \cdot \bv_i \geq t] = \frac{\Pr[\{i, j\} \text{ included in } M ]}{\tail(t)} = \mu_j(i),
    \end{equation}
    since $\{i, j\}$ being included in $M$ implies both $\br \cdot \bv_i \geq t$ and $\br \cdot \bv_j \geq t$. This allows us to think of $G$ as a weighted graph whose weight on the edge $\{i, j\}$ is $\mu_{i}(j) = \mu_j(i)$, which turns out to be a crucial property that the following analysis relies on. We note that for the original KMS algorithm there doesn't seem to be an easy way to define such a symmetric packing.
\end{example}
\subsection{Pruning Two-Step Walks}

\begin{definition}[cf. Definition 3.6.2 in~\cite{chlamtac_non-local_2009}]\label{def:spread}
    Let $\lambda, p > 0$ and $V$ be a set of unit vectors. We say that an $(s, \delta)$-packing $(X, \mu)$ is $(\lambda, p, V)$-spread\footnote{Chlamtáč only defined the notion of $(\lambda, p)$-spread, which is equivalent to $(\lambda, p, X)$-spread in our definition.} if for every $\bv \in V$,
    \begin{equation}
        \mu(\{\bx \in X \mid \bv \cdot \bx \geq \lambda\}) \leq p.
    \end{equation}
\end{definition}
Intuitively, Definition~\ref{def:spread} says that along any direction in $V$, the set $X$ is not too clustered with respect to the measure $\mu$. We expect this to be true for efficient covers no matter what direction we pick. This is quantified in the following lemma.

\begin{lemma}[Corollary 3.6.4 in~\cite{chlamtac_non-local_2009}]\label{lem:c-ineff_spread}
    Let $c > 0$, $V$ be an arbitrary set of unit vectors, and $(X, \mu)$ be an $(s, \delta)$-cover that is at most $c$-inefficient. Then there exists some $c' > 0$ depending only on $c$ such that for all sufficiently large $s$ and every $\epsilon \geq \frac{\log s}{s}$, $(X, \mu)$ is $\left(\sqrt{\frac{c}{1+c}\cdot(1+\epsilon)}, \exp(-c'\epsilon^2s^2),V\right)$-spread.\footnote{Corollary 3.6.4 in~\cite{chlamtac_non-local_2009} is stated for every $s, \epsilon > 0$ and with $f(s)\cdot \exp(-c'\epsilon^2s^2)$ in place of $\exp(-c'\epsilon^2s^2)$, where $f(s)$ is some rational function in $s$. When $s$ is sufficiently large and $\epsilon = \Omega\left(\frac{\log s}{s}\right)$, $f(s)$ is dominated by the exponential term.}
\end{lemma}

Lemma~\ref{lem:c-ineff_spread} is essentially tight, as can be seen by setting $\bv_i = \sqrt{\frac{c}{1+c}}\bv_0 + \sqrt{\frac{1}{1+c}}\bv_i^\perp$ where $\bv_0, \{\bv_i^\perp\}_i$ are mutually orthogonal unit vectors. However, if we are willing to ignore a small fraction of directions (such as the $\bv_0$ direction in the above construction), then we can in fact ``boost'' the parameters in Lemma~\ref{lem:c-ineff_spread}. The following lemma provides one step in ``boosting''.

\begin{lemma}[Lemma 3.6.8 in~\cite{chlamtac_non-local_2009}]\label{lem:spread_boost}
    Let $c > 0$ and $(X, \mu)$ be a $c$-inefficient  $(s, \delta)$-packing which is also $(\lambda, p, X)$-spread. Then, assuming $s$ is sufficiently large\footnote{The assumption that $s$ is sufficiently large was not in the original theorem statement by Chlamtáč. We added this assumption in order to make the parameters easier to present. Since $s$ grows with $n$ in the coloring scenario, this assumption can always be satisfied.}, for every $\sigma \in \mathbb{Z}^+$, $\epsilon \geq \frac{\log s}{s} $ and $\alpha > \frac{e}{2}\left( \exp(-c'\epsilon^2 s^2)\right)^{1/\sigma}$ (where $c' > 0$ is the constant in Lemma~\ref{lem:c-ineff_spread}), there exists $X' \subseteq X$ such that the following properties hold:
    \begin{itemize}
        \item $\mu(X') \geq  \delta - p \cdot \sigma / \alpha$.
        \item $(X', \mu|_{X'})$ is $(\lambda', p', X')$-spread, where $\lambda' = \sqrt{\lambda \cdot \frac{c}{1 + c} (1 + \epsilon)(1 + \frac{1}{\lambda\sigma})}$ and $p' = 2\alpha$.
    \end{itemize}
\end{lemma}

We will state and prove a slightly more general version of Lemma~\ref{lem:spread_boost} in the next section (see Lemma~\ref{lem:spread_boost_X1X2}).
By applying Lemma~\ref{lem:spread_boost} iteratively, Chlamtáč obtained the following theorem.
\begin{theorem}[Theorem 3.6.9 in~\cite{chlamtac_non-local_2009}]\label{thm:pruning_boosted}
    Let $(X, \mu)$ be a $c$-inefficient $(s, \delta)$-packing. If $s$ is sufficiently large, then there exists $X' \subseteq X$ such that the following properties hold:
    \begin{enumerate}
        \item[(1)] $\mu(X') \geq \delta - \frac{1}{\log s}$,
        \item[(2)] $(X', \mu)$ is $(\lambda, p, X')$-spread, where $\lambda = \frac{c}{1+c}\left(1 + \frac{C_1}{\log s}\right)$ and $p = \exp(-C_2 \cdot \log^2s)$, where $C_1, C_2 > 0$ are some constants depending only on $c$.
    \end{enumerate}
\end{theorem}

Finally, we are ready to present the main statement for pruning two-step walks. The version we are stating here has slightly improved parameters compared to Chlamtáč's original statement in~\cite{chlamtac_non-local_2009}, which will be crucial in the next section. For completeness, we include a proof in Appendix~\ref{subsec:thm_chlamtac_pruning_result}.

\begin{theorem}[cf. Lemma 3.6.11 in~\cite{chlamtac_non-local_2009}]\label{thm:chlamtac_pruning_result}
     Fix $c > 0$, and let $n$ be sufficiently large. Assume that $\KMSp$ fails on some $n$-vertex graph $G$ with $c$-inefficient parameter $t$. For every $i \in V$, let $\mu_i$ be defined as in Example~\ref{example:mu_i_p_i}. Then there is a non-empty subgraph $G' = (V', E')$ such that the following properties are satisfied by every $i \in V'$:
    \begin{enumerate}[(1)]
        \item $\mu_i(\{\bv_{ij} \mid j \in N_{G'}(i)\}) \geq \frac{1}{8} - O\left(\frac{1}{\log \log t}\right)$.
        \item For every $j \in N_{G'}(i)$, $(W_{ij}, \mu_j)$ is a $(\sqrt{3}t, \frac{1}{2(\log\log t)^2})$-packing, where 
        \begin{equation}
            W_{ij} \coloneqq \left\{\bv_{jk} \,\Big|\, k \in N_{G'}(j), -\frac{1}{\log \log t} \leq \bv_{ji} \cdot \bv_{jk} \leq \frac{c}{1+c}\left(1 + \frac{C}{\log t}\right)\right\}.
        \end{equation}
    \end{enumerate}
    Here $C > 0$ is some constant depending only on $c$.
\end{theorem}

\subsection{Cover Composition and Finding a Large Vector 2-Colorable Set}

We now compose the two covers obtained by taking a 2-step walk. An essential tool here is a cover composition lemma, which we will also need in the next section. The following Gaussian isoperimetry inequality by Borell is needed in the proof of the cover composition lemma. For any $A \subseteq \mathbb{R}^n$ and $a \geq 0$, let $A + a \coloneqq \{\bu \in \mathbb{R}^n: \exists \bv \in A, \|\bu - \bv\| \leq a\}$.
\begin{theorem}[\cite{borell1985geometric}]\label{thm:borell}
    Let $\gamma_n$ be the standard Gaussian measure over $\mathbb{R}^n$ and $A \subseteq \mathbb{R}^n$ be a measurable set such that $\gamma_n(A) = \Phi(s)$ for some $s \in \mathbb{R}$. Then for any $a \geq 0$, $\gamma_n(A + a) \geq \Phi(s + a)$.
\end{theorem}

\begin{lemma}[Theorem 3.4.4 in~\cite{chlamtac_non-local_2009}, Cover Composition]\label{lem:cover_composition}
    Let $I$ and $\{J_i \mid i \in I\}$ be some index sets. Let $X = \{\bv_i \mid i \in I\}$ be a $c$-inefficient $(s_1, \delta_1)$-cover, and for every $i \in I$, let $Y_i = \{\bu_{ij} \mid j \in J_i\}$ be a set of (not necessarily unit) vectors such that $\bu_{ij} \perp \bv_i$ for every $j \in J_i$, and $\Pr_\br[\exists j \in J_i: \br \cdot \bu_{ij} \geq s_2] \geq \delta_2$ for some $s_2, \delta_2 > 0$. Then we have 
    \begin{equation}
        \Pr_\br\left[\exists i \in I, j \in J_i: \br \cdot \bv_i \geq s_1, \br \cdot \bu_{ij} \geq s_2  - a\|\bu_{ij}\|\right] \geq \delta_1 - \exp(-\Omega_c(s_1)),
    \end{equation}
    where $a = \sqrt{c(1+1/s_1)}\cdot s_1 - \Phi^{-1}(\delta_2)$. \footnote{In our setting, we have $s_1, s_2 = \Theta(\sqrt{\log n})$ and $\delta_1, \delta_2 = \widetilde{\Omega}(1)$, where $n$ is the number of vertices in the graph. So the dominating term in $a$ is $\sqrt{c}\cdot s_1 $, and the remaining terms are negligible. The $\exp(-\Omega_c(s_1))$ term on the right hand side of the inequality is also negligible.}
\end{lemma}
\begin{proof}
    Fix some $i \in I$. By Theorem~\ref{thm:borell}, for every $a \geq 0$ we have 
    \begin{equation}
        \Pr_\br\left[\exists j \in J_i: \br \cdot \bu_{ij} \geq s_2  - a \|\bu_{ij}\|\right] \geq \Phi(\Phi^{-1}(\delta_2) + a).
    \end{equation}
    Taking $a = \sqrt{c(1+1/s_1)}\cdot s_1 - \Phi^{-1}(\delta_2)$, we have
    \begin{align}
        & \Pr_\br\left[ \br \cdot \bv_i \geq s_1 \text{ and }\forall j \in J_i: \br \cdot \bu_{ij}\leq s_2 - a \|\bu_{ij}\|\right] \\
        =\, & \Pr_\br\left[ \br \cdot \bv_i \geq s_1\right] \cdot\Pr_\br\left[\forall j \in J_i: \br \cdot \bu_{ij} \leq s_2 - a \|\bu_{ij}\|\right]\\
        \leq\, & \tail(s_1)\cdot\tail(\sqrt{c(1+1/s_1)}s_1).
    \end{align}
    It follows that 
    \begin{align}
         & \Pr_\br\left[\exists i \in I, j \in J_i: \br \cdot \bv_i \geq s_1, \br \cdot \bu_{ij} \geq s_2 - a \|\bu_{ij}\|\right] \\ 
        \geq \, & \Pr_\br[\exists i \in I: \br \cdot \bv_i \geq s_1] - \sum_{i \in I} \Pr_\br\left[ \br \cdot \bv_i \geq s_1 \text{ and }\forall j \in J_i: \br \cdot \bu_{ij} \leq s_2 - a \|\bu_{ij}\|\right] \\
        \geq \, & \delta_1 - |X| \cdot \tail(s_1)\cdot\tail(\sqrt{c(1+1/s_1)}s_1) \\
        \geq \, & \delta_1 - \tail(s_1)^{-(1+c)} \cdot \tail(s_1)\cdot\tail(\sqrt{c(1+1/s_1)}s_1) \\
        \geq \, & \delta_1 - \exp(-\Omega_c(s_1)).
    \end{align}
    Here, in the last two inequalities, we used the $c$-inefficiency of $X$ and Fact~\ref{fact:gaussian_tail}.
\end{proof}

We also need the following simple lemma which shows that projecting a cover in some fixed direction does not hurt the probability too much.

\begin{lemma}[Lemma 3.4.2 in~\cite{chlamtac_non-local_2009}]\label{lem:v_i-v_0}
    Let $\bv_0$ be a unit vector and $\{\bv_i \mid i \in I\}$ be a $(s, \delta)$-cover. Then for every $\rho \geq 0$, we have 
    \begin{equation}
        \Pr_\br[\exists i: \br \cdot (\bv_i - (\bv_i \cdot \bv_0) \bv_0) \geq s - \rho] \geq \delta - 2\tail(\rho).
    \end{equation}
\end{lemma}
In our applications of the above lemma we will have $s = \Theta(t)$ and $\delta = \Omega(\frac{1}{(\log\log t)^2})$, so we may take $\rho = \log t$ which is negligible. We now apply the cover composition lemma and analyze the quality of the two-step cover. We state it in a way that is also applicable for the purpose of the next section.

\begin{theorem}[cf. Theorem 3.4.5 in~\cite{chlamtac_non-local_2009}]\label{thm:two_steps_main}
    Fix $c > 0$, and let $n$ be sufficiently large. Assume that $\KMSp$ fails on some $n$-vertex graph $G$ with $c$-inefficient parameter $t$. Let $G' = (V', E')$, $\{\mu_{i} \mid i \in V'\}$, $\{W_{ij} \mid i \in V', j \in N_{G'}(i)\}$ be as in Theorem~\ref{thm:chlamtac_pruning_result}. Let $i \in V'$ and $\delta_1, \delta_2 = \Omega(\frac{1}{(\log \log t)^2})$. Assume that $S \subseteq N_{G'}(i)$, $\{S_j \subseteq W_{ij} \mid j \in S\}$ are some sets of vertices satisfying 
    \begin{itemize}
        \item[(1)] $\mu_i(S) \geq \delta_1$;
        \item[(2)] for every $j \in S$, $\mu_j(S_j) \geq \delta_2$.
    \end{itemize}
    Then we have 
 \[
        \Pr_\br\left[\exists k \in \bigcup_{j \in S} S_j: \br \cdot \bv_{ik} \geq \frac{9-3\alpha_k-6\sqrt{(1 - \alpha_k^2)c}}{4\sqrt{1 - \left(\frac{1}{4} + \frac{3}{4}\alpha_k\right)^2}} \cdot t  - O(\log t)\right] \geq \frac{\delta_1}{2}
    \]
    where $\alpha_k = \bv_{ji} \cdot \bv_{jk}$ for $k \in S_j$. In particular, if 
    \[
    \eta_0 = \inf_{0 \leq \alpha \leq \frac{c}{1+c}}\frac{9-3\alpha-6\sqrt{(1 - \alpha^2)c}}{4\sqrt{1 - \left(\frac{1}{4} + \frac{3}{4}\alpha\right)^2}},
    \]
    then     \begin{equation}\label{eq:second_level_size}
        \left| \bigcup_{j \in S} S_j\right| = \Omega\left(\Delta(G)^{\frac{\eta^2}{3(1+c)}}\right).
    \end{equation}
    for any $\eta \in (0, \eta_0)$.
\end{theorem}
\begin{proof}
    For any $j \in S$ and $k \in S_j$, we can write
    \begin{align}
        \bv_k & = \frac{1}{4}\bv_i - \frac{\sqrt{3}}{4}\bv_{ij} + \frac{\sqrt{3}}{2}\bv_{jk} \\
        & = \left(\frac{1}{4} + \frac{3}{4}\alpha_k\right)\bv_i - \frac{\sqrt{3}}{4}(1 - \alpha_k)\bv_{ij} + \frac{\sqrt{3}}{2}(\bv_{jk} - \alpha_k \bv_{ji}), \label{eq:two_step_decomp}
    \end{align}
    where $\alpha_k = \bv_{jk} \cdot \bv_{ji}$. Note that $\bv_{jk} - \alpha_k \bv_{ji}$ is orthogonal to $\bv_j$ and $\bv_{ji}$, so it is orthogonal to $\bv_i$ and $\bv_{ij}$ as well. This means $\bv_i \cdot \bv_k = \frac{1}{4} + \frac{3}{4}\alpha_k$, which justifies that the definition of $\alpha_k$ does not depend on $j$. Now let us apply Lemmas~\ref{lem:cover_composition} and~\ref{lem:v_i-v_0} to the cover $\{-\bv_{ij} \mid j \in S\}$ and the sets of vectors $\{\{\bv_{jk} - \alpha_k \bv_{ji} \mid k \in S_j\} \mid j \in S\}$, which gives us 
    \begin{align}\label{eq:two_step_composition}
        \Pr\Big[\exists j \in S, k \in S_j:\, & \br \cdot (-\bv_{ij}) \geq \sqrt{3}t\, \text{ and } \\
        & \br \cdot (\bv_{jk} - \alpha_k \bv_{ji}) \geq (1 - \sqrt{c}\|\bv_{jk} - \alpha_k \bv_{ji}\|)\sqrt{3}t - O(\log t)\Big] \geq \frac{\delta_1}{2}
    \end{align}
    When the event on the left hand side above happens, we have
    \begin{align}
        \br \cdot \left( \bv_k - \left(\frac{1}{4} + \frac{3}{4}\alpha_k\right)\bv_i \right) =\, &\,\br \cdot \left(- \frac{\sqrt{3}}{4}(1 - \alpha_k)\bv_{ij} + \frac{\sqrt{3}}{2}(\bv_{jk} - \alpha_k \bv_{ji})\right) \\
        \geq\, &\, \frac{3(1 - \alpha_k)}{4}t + \frac{1 - \sqrt{c}\sqrt{1 - \alpha_k^2}}{2}\cdot 3t  - O(\log t) \\
        =\, &\, \left(\frac{9 - 3\alpha_k - 6\sqrt{(1-\alpha_k^2)c}}{4}\right) t  - O(\log t) .
    \end{align}
    So it follows from \eqref{eq:two_step_composition} that
    \begin{equation}
        \Pr\left[\exists k \in \bigcup_{j \in S} S_j: \br \cdot \bv_{ik} \geq \left(\frac{9 - 3\alpha_k - 6\sqrt{(1-\alpha_k^2)c}}{4\sqrt{1 - \left(\frac{1}{4} + \frac{3}{4}\alpha_k\right)^2}}\right) t  - O(\log t)\right] \geq \frac{\delta_1}{2}.
    \end{equation}
    
    If $\eta_0 = \inf_{0 \leq \alpha \leq \frac{c}{1+c}}\frac{9-3\alpha-6\sqrt{(1 - \alpha^2)c}}{4\sqrt{1 - \left(\frac{1}{4} + \frac{3}{4}\alpha\right)^2}}$, then 
    \begin{equation}
        \Pr_\br\left[\exists k \in \bigcup_{j \in S} S_j: \br \cdot \bv_{ik} \geq \eta_0 \cdot t  - o(t)\right] \geq \frac{\delta_1}{2}.
    \end{equation}
    (Here $O(\log t)$ becomes $o(t)$ due to ignored $o(1)$ terms in the range of $\alpha$.) By Fact~\ref{fact:cover_lower_bd}, we have 
    \begin{equation}
        \left| \bigcup_{j \in S} S_j\right| = \widetilde{\Omega}\left(\tail(\eta_0 \cdot t   - o(t))^{-1}\right).
    \end{equation}
    Since $t$ is $c$-inefficient, $\tail(\sqrt{3}t)^{-(1+c)} = \Delta(G)$, and therefore \eqref{eq:second_level_size} follows from Fact~\ref{fact:gaussian_tail}.
\end{proof}

\begin{lemma}\label{lemma:second_level_vector_coloring}
        Fix $c > 0$, and let $n$ be sufficiently large. Assume that $\KMSp$ fails on some $n$-vertex graph $G$ with $c$-inefficient parameter $t$. Let $G' = (V', E')$, $\{\mu_{i} \mid i \in V'\}$, $\{W_{ij} \mid i \in V', j \in N_{G'}(i)\}$ be as in Theorem~\ref{thm:chlamtac_pruning_result}. Then, for any $\epsilon > 0$, there exists a vector $(2 + \epsilon)$-coloring on the subgraph of $G$ induced by $\cup_{j \in N_{G'}(i)}W_{ij}$.
\end{lemma}
\begin{proof}
    By \eqref{eq:two_step_decomp}, for every $k \in \cup_{j \in N_{G'}(i)}W_{ij}$ we have $\frac{1}{4} - \frac{1}{\log \log t} \leq \bv_i \cdot \bv_k \leq \frac{1}{4} + \frac{3}{4}\cdot \frac{c}{1+c} + \frac{C}{\log t}$. By Lemma~\ref{lem:vector_coloring_positive}, we obtain a vector $(2 + \epsilon)$-coloring on the subgraph of $G$ induced by $\cup_{j \in N_{G'}(i)}W_{ij}$  for any $\epsilon > 0$.
\end{proof}

Theorem~\ref{thm:chlamtac_clean_statement} is now a simple corollary.

\begin{proof}[Proof of Theorem~\ref{thm:chlamtac_clean_statement}]
    Taking $S = N_{G'}(i)$ and $S_j = W_{ij}$ for every $j \in N_{G'}(i)$ in Theorem~\ref{thm:two_steps_main}, we get that for any $\eta \in (0, \eta_0)$
    \begin{equation}
        \left| \bigcup_{j \in N_{G'}(i)} W_{ij}\right| = \Omega\left(\Delta(G)^{\frac{\eta^2}{3(1+c)}}\right).
    \end{equation}

    Applying the $\KMS$ algorithm to the $(2 + o(1))$-coloring given by Lemma~\ref{lemma:second_level_vector_coloring}, we obtain an independent set of size $\Omega\left(\Delta(G)^{\frac{\eta^2}{3(1+c)}}\right)$.
\end{proof}

\section{Analyzing Third-Level Neighborhoods}\label{sec:three_steps}

In this section, we extend the analysis in the previous section to third-level neighborhoods. The overall idea is to show that there exists some vertex $i$ whose third-level neighborhood $N^{(3)}(i)$ is sufficiently ``random-like'', in the sense that for a ``generic'' $\ell \in N^{(3)}(i)$, we have $\bv_i \cdot \bv_\ell \approx (-1/2)^3 = -1/8$. Using Lemma~\ref{lem:vector_coloring_negative}, we can obtain a vector coloring with parameter close to $5/2$ of these vertices, on which we can then apply the KMS algorithm and obtain an independent set of size roughly $\Omega(|N^{(3)}(i)| \cdot \Delta(G)^{-1/5})$. This plan of attack comes with two technical challenges:
\begin{itemize}
    \item It is crucial to bound how much $\bv_i \cdot \bv_\ell$ actually differs from $-1/8$. Using a naive bound without pruning, one can show that $\bv_i \cdot \bv_\ell = -1/8 + \Theta(\sqrt{c})$ for a ``generic'' $\ell \in N^{(3)}(i)$. However, this loss translates poorly to the loss in vector chromatic number that we obtain from Lemma~\ref{lem:vector_coloring_negative}. Indeed, even for mildly big $c = 0.01$, the $\Theta(\sqrt{c})$ term can easily push the inner product to $0$, in which case the vector coloring we obtain reduces to a vector 3-coloring and no progress can be made. It is therefore essential to extend the pruning technique from two-step walks to three-step walks, which will give us $\bv_i \cdot \bv_\ell = -1/8 + \Theta(c)$. This is the focus of Section~\ref{subsec:more_pruning}.
    
    \item We also want the number of such ``generic'' vertices in $N^{(3)}(i)$ to be large. In fact, we need it to be much bigger than the second-level neighborhood, in order to overcome the loss from the term $\Delta(G)^{-1/5}$ which is introduced because we go from a vector 2-coloring on the second-level neighborhood to a vector $\frac{5}{2}$-coloring on the third-level. This analysis requires us to compose the two-step cover on the second-level neighborhood with an extra step, using the cover composition lemma (Lemma~\ref{lem:cover_composition}). However, the loss term in the cover composition lemma depends on the efficiency of the two-step cover, and if it is very inefficient, then we may not get any improvement at all! This might look bad at first glance, but a moment of thought reveals that we are in fact in a win-win situation, namely, if the two-step cover is very inefficient, then it means it has a lot more vertices than the estimate in Theorem~\ref{thm:chlamtac_clean_statement}, which implies that we already have an improvement; otherwise, it is efficient enough and we do make meaningful progress in the third-level neighborhood. This analysis is carried out in detail in Section~\ref{subsec:win-win}.
    
\end{itemize}

We introduce the parameters appearing in our main theorem. For an $n$-vertex graph $G$, write $\Delta = \Delta(G)$. For $0<c<1/5$ and $c'>0$, define
\[
        f_{c'}(n, \eta) = \Delta^{\frac{\eta^2(1 + c')}{3(1+c)}}, \quad g_{c'}(n, \lambda) =  \Delta^{\frac{\lambda^2}{3(1+c)} -\frac{1+3c}{5-c}}.
\]
Here the dependence on $c$ and $\Delta(G)$ is suppressed in the notation. Recall the parameter
\[
    \eta_0 = \eta_0(c)
    = \inf_{\alpha \in [0,c/(1+c)]}
    \frac{9-3\alpha-6\sqrt{(1-\alpha^2)c}}
         {4\sqrt{1-\left(\frac14+\frac34\alpha\right)^2}}
\]
from Theorem~\ref{thm:chlamtac_clean_statement}, and define
\[
    \lambda_0 = \lambda_0(c,c')
    =
    \inf_{\substack{
        \beta \in [\frac14,\frac14+\frac{3c}{4(1+c)}] \\
        \gamma \in [-\frac{3\sqrt{3}c}{4(1+c)},
                     \frac{\sqrt{3}c}{2(1+c)}]
    }}
    \frac{
        \left(
            \frac12\sqrt{1-\beta^2}
            + \frac{\sqrt{3}}2
              \frac{\beta\gamma}{\sqrt{1-\beta^2}}
            - \frac{\sqrt{3}}2\sqrt{c'}
              \sqrt{1-\frac{\gamma^2}{1-\beta^2}}
        \right)\eta_0+\frac32
    }{
        \sqrt{1-\left(-\frac12\beta
                      +\frac{\sqrt{3}}2\gamma\right)^2}
    }.
\]
The functions $f_{c'}$ and $g_{c'}$ correspond to the two cases of the win-win argument described above: a sufficiently large second-level cover yields the guarantee $f_{c'}$, while an at most $c'$-inefficient second-level cover allows us to extend the analysis to the third level and obtain the guarantee $g_{c'}$.

\begin{theorem}\label{thm:our_clean_statement}
    Fix $0<c<1/5$ and $c' > 0$, and let $n$ be sufficiently large. Assume that $\KMSp$ fails on an $n$-vertex graph $G$ with $c$-inefficient parameter $t$. Then, for any $\eta \in (0,\eta_0)$ and $\lambda \in (0,\lambda_0)$, we can find an independent set in $G$ of size
    \[
        \Omega\left(
            \min\{f_{c'}(n,\eta),g_{c'}(n,\lambda)\}
        \right).
    \]
\end{theorem}

\subsection{More Pruning}\label{subsec:more_pruning}
Using the analysis from the previous section, we may obtain 2-step walks $i \to j \to k$ where $\bv_{ij} \cdot \bv_{jk} = O(c)$. By applying the same analysis twice, we get 3-step walks $i \to j \to k \to \ell$ for which both $\bv_{ij} \cdot \bv_{jk} = O(c)$ and $\bv_{jk} \cdot \bv_{k\ell} = O(c)$. However, this is not sufficient to guarantee that $\bv_{k\ell}$ and $\bv_{ij}$ also have small inner product. We now apply more pruning to make sure they do. The first lemma we prove is a small generalization of Lemma~\ref{lem:spread_boost} to the case where we prune one packing against some other packing. The proof given here largely follows Chlamtáč's argument in~\cite{chlamtac_non-local_2009}, and the only difference is in distinguishing vectors from these two different packings. For the proof we need the following combinatorial lemma from Chlamtáč.

\begin{lemma}[Corollary 3.6.7 in~\cite{chlamtac_non-local_2009}]\label{lem:intersection}
    Let $\mu$ be a measure on some finite set $X$ such that $\mu(X) \leq 1$. Let $X_1, \ldots, X_k$ be subsets of $X$ such that $\sum_{i=1}^k\mu(X_i) \geq 2 \ell$ for some $\ell \in \mathbb{Z}^+$. Then there exists a subset $S \subseteq [k]$ of size $\ell$ such that 
    \begin{equation}
        \mu\left(\bigcap_{i \in S} X_i\right) \geq \left(\frac{4\ell}{ek}\right)^\ell.
    \end{equation}
\end{lemma}

\begin{lemma}[cf. Lemma~\ref{lem:spread_boost}]\label{lem:spread_boost_X1X2}
    Let $(X_1, \mu_1)$ be an $(s_1, \delta_1)$-packing and $(X_2, \mu_2)$ be an $(s_2, \delta_2)$-packing. Furthermore, assume that $X_2$ is $c$-inefficient and $(X_1, \mu_1)$ is $(\lambda, p, X_1)$-spread. Then, assuming $s_2$ is sufficiently large, for every $\sigma \in \mathbb{Z}^+$, $\epsilon \geq \frac{\log s_2}{s_2}$ and $\alpha > \frac{e}{4}\left( \exp(-c'\epsilon^2 s^2_2)\right)^{1/\sigma}$ (where $c' > 0$ is the constant in Lemma~\ref{lem:c-ineff_spread}), there exists $X_1' \subseteq X_1$ such that the following properties hold:
    \begin{itemize}
        \item $(X_1', \mu_1)$ is an $(s_1, \delta_1 - p \cdot \sigma / \alpha)$-packing.
        \item $(X_2, \mu_2)$ is $(\lambda', p', X_1')$-spread, where $\lambda' = \sqrt{\lambda \cdot \frac{c}{1 + c} (1 + \epsilon)(1 + \frac{1}{\lambda\sigma})}$ and $p' = 2\alpha$.
    \end{itemize}
\end{lemma}

\begin{proof}
    To avoid confusion, let us use $\bu, \bu'$ to denote vectors in $X_1$ and $\bv, \bv'$ to denote vectors in $X_2$. For every $\bu \in X_1$, let $Y_\bu \coloneqq \{\bv \in X_2 \mid \bu \cdot \bv \geq \lambda'\}$. Consider the  procedure where we do the following repeatedly: if $\mu_2(Y_\bu) \geq 2\alpha$ for some $\bu \in X_1$, then we remove every $\bu' \in X_1$ from $X_1$ such that $\bu \cdot \bu' \geq \lambda$. We also add the vector $\bu$ to some initially empty set $T$. 
    Since $(X_1, \mu_1)$ is $(\lambda, p, X_1)$-spread, each iteration removes a set of measure (w.r.t. $\mu_1$) at most $p$ from $X_1$. We now show by contradiction that this procedure terminates within $\sigma/\alpha$ steps. If this process terminates within $\sigma / \alpha$ steps, then let $X_1'$ be the remaining vectors in $X_1$; it has the desired property. Otherwise, we obtain a set $T$ of size $\frac{\sigma}{\alpha}$ and $\sum_{\bu \in T}\mu_2(Y_{\bu}) \geq \frac{\sigma}{\alpha} \cdot 2\alpha = 2\sigma$. By Lemma~\ref{lem:intersection} (taking $k = \frac{\sigma}{\alpha}$ and $\ell = \sigma$), we can find some $S \subseteq T$ such that $|S| = \sigma$ and 
    \begin{equation}\label{eq:mu_2_large}
        \mu_2\left(\bigcap_{\bu \in S} Y_\bu\right) \geq \left(\frac{4\alpha}{e}\right)^\sigma > \exp(-c'\epsilon^2 s_2^2).
    \end{equation}
    Let $\bu_0 = \sum_{\bu \in S} \bu / \left\|\sum_{\bu \in S}\bu\right\|$. By construction, $\bu \cdot \bu' < \lambda$ for every $\bu, \bu' \in S$, so 
    \begin{equation}
        \left\|\sum_{\bu \in S}\bu\right\| < \sqrt{|S| + |S|(|S| - 1)\lambda} < \sqrt{\sigma + \sigma^2\lambda}.
    \end{equation}
    It follows that for every $\bv \in \bigcap_{\bu \in S} Y_\bu$,
    \begin{equation}
        \bu_0 \cdot \bv \geq \frac{1}{\sqrt{\sigma + \sigma^2\lambda}} \cdot \sum_{\bu \in S} \bu \cdot \bv \geq \frac{\sigma \lambda'}{\sqrt{\sigma + \sigma^2\lambda}} = \frac{\sigma \sqrt{\lambda \cdot \frac{c}{1 + c} (1 + \epsilon)(1 + \frac{1}{\lambda\sigma})}}{\sqrt{\sigma + \sigma^2\lambda}} = \sqrt{\frac{c}{1 + c} (1 + \epsilon)}.
    \end{equation}
    Since $X_2$ is $c$-inefficient, by Lemma~\ref{lem:c-ineff_spread}, this contradicts \eqref{eq:mu_2_large}.
\end{proof}

For a set of vectors $X$, write
\[
    -X = \{-\bv \mid \bv\in X\},
    \qquad
    \overline{X} = X\cup(-X).
\]
Thus, $\overline{X}$ is the symmetrization of $X$.

\begin{corollary}\label{cor:X1X2prune}
    For $i = 1, 2$, let $(X_i, \mu_i)$ be a $c$-inefficient $(s, \delta_i)$-packing. Assume that $s$ is sufficiently large. Then, there exists $X_1' \subseteq X_1$ such that the following properties hold:
    \begin{itemize}
        \item $(X_1', \mu_1)$ is an $(s, \delta_1 - O(1/\log s))$-packing.
        \item $(X_2, \mu_2)$ is $(\lambda, p, \overline{X_1'})$-spread, where $\lambda = \frac{c}{1+c}\left(1 + \frac{c_1}{\log s}\right)$ and $p \leq \exp(-c_2 \cdot \log^2s)$ for some constants $c_1, c_2 > 0$ depending only on $c$.
    \end{itemize}
\end{corollary}

\begin{proof}
    Let $Y = \overline{X_1}$. We define a measure $\mu_Y$ over $Y$ as follows: for each $\bv \in Y$, let $\mu_Y(\bv) = \frac{\mu_1(\bv) + \mu_1(-\bv)}{2}$, where we take $\mu_1(\bu) = 0$ for any $\bu \notin X_1$. Then $(Y, \mu_Y)$ is also an $(s, \delta_1)$-packing. Indeed, for any $S \subseteq Y$, we have 
    \begin{equation}
        \mu_Y(S) = \sum_{\bv \in S} \frac{\mu_1(\bv) + \mu_1(-\bv)}{2} = \frac{\mu_1(S \cap X_1) + \mu_1((-S) \cap X_1)}{2}.
    \end{equation}
    So $\mu_Y(Y) = \mu_1(X_1) \geq \delta_1$, and 
    \begin{align}
        \Pr_\br[\exists v \in S : \br \cdot \bv \geq s] & = \frac{\Pr_\br[\exists v \in S  : \br \cdot \bv \geq s] + \Pr_\br[\exists v \in S  : \br \cdot \bv \geq s]}{2} \\ & \geq \frac{\Pr_\br[\exists v \in S \cap X_1 : \br \cdot \bv \geq s] + \Pr_\br[\exists v \in S \cap (-X_1) : \br \cdot \bv \geq s]}{2} \\
        & \geq \frac{\Pr_\br[\exists v \in S \cap X_1 : \br \cdot \bv \geq s] + \Pr_\br[\exists v \in (-S) \cap X_1 : \br \cdot \bv \geq s]}{2} \\
        & \geq \frac{\mu_1(S \cap X_1) + \mu_1((-S) \cap X_1 )}{2} = \mu_Y(S).
    \end{align}

    Since $|Y| \leq 2|X_1|$, $(Y, \mu_Y)$ is $(c + o(1))$-inefficient, and therefore we may apply Theorem~\ref{thm:pruning_boosted} to $(Y, \mu_Y)$ and obtain a subset $Y' \subseteq Y$ such that: 
    \begin{itemize}
        \item $\mu_Y(Y') \geq \mu_Y(Y) - 1/\log s$;
        \item $(Y', \mu_Y)$ is $(\lambda_1, p_1, Y')$-spread for some $\lambda_1 = \frac{c}{1+c}(1 + \frac{C_1}{\log s})$ and $p_1 = \exp(-C_2\log^2s)$ where $C_1, C_2 > 0$ are constants depending only on $c$.
    \end{itemize}
    Since $\mu_Y(\bv) = \mu_Y(-\bv)$ for any $\bv \in Y$, we have 
    \begin{equation}
        \mu_Y(Y' \cap (-Y')) \geq \mu_Y(Y)- \frac{2}{\log s} \geq \delta_1 - \frac{2}{\log s}.
    \end{equation} 
    Since $Y' \cap (-Y')$ is a subset of $Y'$, we have that $(Y' \cap (-Y'), \mu_Y)$ is also $(\lambda_1, p_1, Y')$-spread, and therefore also $(\lambda_1, p_1, Y' \cap (-Y'))$-spread.

    We now apply Lemma~\ref{lem:spread_boost_X1X2} to $(Y' \cap (-Y'), \mu_Y)$ and $(X_2, \mu_2)$, with $\sigma = \lfloor \log s \rfloor$, $\lambda = \lambda_1$, $p = p_1$, $\alpha = 2 p \sigma \log s$ and $\epsilon = \frac{\log^{2} s}{s}$. Note that for any $c' > 0$ we have $\left( \exp(-c'\epsilon^2 s^2)\right)^{1/\sigma} \leq \exp(-c'\log^3 s)$, whereas $\alpha = 2p\sigma \log s > (\log^2s)\cdot\exp(-C_2\log^2s)$ so $\alpha > \frac{e}{4}\left( \exp(-c'\epsilon^2 s^2)\right)^{1/\sigma}$ when $s$ is sufficiently large as required by Lemma~\ref{lem:spread_boost_X1X2}. This gives us $Y'' \subseteq Y$ such that 
    \begin{equation}
        \mu_Y(Y'') \geq \mu_Y(Y' \cap (-Y')) - \frac{p\sigma}{\alpha} \geq \delta_1 - O\left(\frac{1}{\log s}\right)
    \end{equation} 
    and furthermore $(X_2, \mu_2)$ is $(\lambda', p', Y'')$-spread, where 
    \begin{align}
        \lambda' & = \sqrt{\lambda_1 \cdot \frac{c}{1 + c} (1 + \epsilon)\left(1 + \frac{1}{\lambda_1\sigma}\right)} \\& = \sqrt{\frac{c}{1+c}\left(1 + \frac{C_1}{\log s}\right) \cdot \frac{c}{1 + c} (1 + \epsilon)\left(1 + \frac{(1 + c)}{(c + o(1))\log s}\right)} \\
        & = \frac{c}{1 + c} \cdot \left(1 + O\left(\frac{1}{\log s}\right)\right)
    \end{align} and $p' = 2\alpha \leq 4p_1\log^2 s = 4\exp(-C_2\log^2s) \log^2s$ which is less than $\exp(-(C_2/2)\log^2s)$ if $s$ is sufficiently large. We complete the proof by taking $X_1' = X_1 \cap (Y'' \cap (-Y''))$.
\end{proof}

For the remainder of this section, we work with the subgraph $G' = (V', E')$ as provided by Theorem~\ref{thm:chlamtac_pruning_result}. For example, we will use $N(i)$ to denote $N_{G'}(i)$, unless explicitly stated otherwise. 

By Theorem~\ref{thm:chlamtac_pruning_result}, every edge $\{j, k\} \in E'$ is associated with two $(\sqrt{3}t, \frac{1}{2(\log\log t)^2})$-packings, namely $(W_{jk}, \mu_k)$ and $(W_{kj}, \mu_j)$, with the property that for every $i \in W_{kj}$ and $\ell \in W_{jk}$, we have $\bv_{ji} \cdot \bv_{jk}, \bv_{kj} \cdot \bv_{k\ell} \in [-o(1), \frac{c}{1+c}(1 + o(1))]$. Applying Corollary~\ref{cor:X1X2prune} to $W_{kj}$ and $W_{jk}$, we can find a subset $W_{kj}' \subseteq W_{kj}$ such that
\begin{itemize}
    \item $\mu_{j}(W_{kj}') \geq \frac{1}{2(\log\log t)^2} - O(1/\log t)$;
    \item For every $i \in W_{kj}'$, 
        \begin{equation}
            \mu_k\left(\left\{ \bv_{k\ell} \;\middle|\; \ell \in W_{jk}, | \bv_{k\ell} \cdot \bv_{ji} | \geq \frac{c}{1+c} \left(1 + \frac{C}{\log t}\right)\right\}\right) \leq \exp(-\Omega(\log^2 t)),
        \end{equation}
        which implies that 
        \begin{equation}
            \mu_k\left(\left\{ \bv_{k\ell} \;\middle|\; \ell \in W_{jk}, | \bv_{k\ell} \cdot \bv_{ji} | \leq \frac{c}{1+c} \left(1 + \frac{C}{\log t}\right)\right\}\right) \geq \frac{1}{2(\log\log t)^2} - \exp(-\Omega(\log^2 t)).
        \end{equation}
\end{itemize} 

For simplicity, let us drop the $o(1)$ terms in the inner products, which do not meaningfully affect our calculations. For $j \in N(i), k \in N(j)$, we say that $k$ is good for $(i, j)$ if $\bv_{ji} \cdot \bv_{jk} \in [0, \frac{c}{1+c}]$ and
    \begin{equation}\label{eq:k_good_def}
        \mu_k\left(\left\{\bv_{k\ell} \;\middle|\; \ell \in N(k), \bv_{k\ell} \cdot \bv_{kj} \in \left[0, \frac{c}{1+c}\right] \text{ and } \bv_{k\ell} \cdot \bv_{ji} \in \left[-\frac{c}{1+c}, \frac{c}{1+c}\right]\right\}\right) \geq \frac{1}{4(\log \log t)^2}.
    \end{equation}
The preceding discussion can be summarized in the following lemma.

\begin{lemma}\label{lem:good_k}
    
     Let $t$ be sufficiently large. Then, for every $j$ and $k \in N(j)$, 
    \begin{equation}
        \mu_j(\{\bv_{ji} \mid i \in N(j), k \text{ good for } (i, j)\}) \geq \frac{1}{4(\log\log t)^2}
    \end{equation}
\end{lemma}

We now aim to ``turn the above lemma around''. More specifically, we will try to find some $i$ such that for many $j \in N(i)$, $\mu_j(\{\bv_{jk} \mid k \in N(j), k \text{ good for } (i, j)\})  = \widetilde{\Omega}(1)$. Let us define $\nu(i, j) = \mu_i(j) \cdot \sum_{k \in N(j), k \text{ good for } (i,j)} \mu_j(k)$.

\begin{lemma}\label{lem:nu_sum}
    Let $t$ be sufficiently large. We have \[
    \sum_{i \in V'}\sum_{j \in N(i)} \nu(i, j) \geq \frac{1}{4(\log\log t)^2}\sum_{i \in V'}\sum_{j \in N(i)} \mu_i(j).
    \]
\end{lemma}
\begin{proof}
    Fix an arbitrary edge $\{j, k\} \in E'$. By Lemma~\ref{lem:good_k}, we have
    \begin{equation}
        \sum_{i \in N(j): k \text{ good for } (i, j)} \mu_j(i) = \mu_j(\{\bv_{ji} \mid i \in N(j), k \text{ good for } (i, j)\}) \geq \frac{1}{4(\log\log t)^2}.
    \end{equation}
    So we have 
    \begin{align}
        \sum_{i \in V'}\sum_{j \in N(i)} \nu(i, j) & =  \sum_{i \in V'}\sum_{j \in N(i)}  \mu_i(j) \cdot \sum_{k \in N(j): k \text{ good for } (i,j)} \mu_j(k) \\
         & =  \sum_{j \in V'}\sum_{k \in N(j)}  \mu_j(k) \cdot \sum_{i \in N(j): k \text{ good for } (i,j)} \mu_j(i) \\
         & \geq \frac{1}{4(\log\log t)^2} \cdot \sum_{j \in V'}\sum_{k \in N(j)}  \mu_j(k).
    \end{align}
    Here we used the fact that $\mu_j(i) = \mu_i(j)$ for every $\{i, j\} \in E'$.
\end{proof}

\begin{lemma}
    Assume that $t$ is sufficiently large. There exists some $i \in V'$ such that
    \begin{equation}\label{eq:good_i}
        \mu_i\left(\left\{\bv_{ij} \mid j \in N(i), \nu(i, j) \geq \frac{\mu_i(j)}{8(\log \log t)^2}\right\}\right) \geq \frac{1}{100(\log\log t)^2}. 
    \end{equation}
\end{lemma}
\begin{proof}
    By Lemma~\ref{lem:nu_sum}, there exists some $i \in V'$ such that 
    \begin{equation}
        \sum_{j \in N(i)} \nu(i, j) \geq \frac{1}{4(\log\log t)^2}\sum_{j \in N(i)}  \mu_i(j).
    \end{equation} 
    Let $S = \left\{ j \in N(i) \mid \nu(i, j) \geq \frac{\mu_i(j)}{8(\log \log t)^2} \right\}$. Then we have 
    \begin{align}
        \frac{1}{4(\log\log t)^2}\sum_{j \in N(i)}  \mu_i(j) & \leq \sum_{j \in N(i)} \nu(i, j) \\
        & \leq \sum_{j \in S} \nu(i, j) + \sum_{j \in N(i) \setminus S} \nu(i, j) \\
        & \leq \sum_{j \in S} \mu_i(j) + \sum_{j \in N(i)} \frac{\mu_i(j)}{8(\log\log t)^2}.
    \end{align}
    Rearranging, we get 
    \begin{equation}
        \sum_{j \in S} \mu_i(j) \geq \frac{1}{8(\log\log t)^2}\sum_{j \in N(i)}  \mu_i(j).
    \end{equation}
    By Theorem~\ref{thm:chlamtac_pruning_result}, we have $\sum_{j \in N(i)}  \mu_i(j) \geq \frac{1}{8} - O\left(\frac{1}{\log \log t}\right)$, so the lemma follows.
\end{proof}

We are now ready to state the main pruning result of this section.

\begin{theorem}\label{thm:our_pruning_main}
    Let $n$ be sufficiently large, and let $G' = (V', E')$ be as in Theorem~\ref{thm:chlamtac_pruning_result}.  Then there exist some $i \in V'$ and $W \subseteq N^{(2)}_{G'}(i)$ satisfying the following:
    \begin{itemize}
        \item[(1)] For every $k \in W$, $\bv_k \cdot \bv_i = \frac{1}{4} + \frac{3}{4}\alpha_{k}$ for some $\alpha_{k} \in [-\frac{1}{\log \log t}, \frac{c}{1+c}+\frac{C}{\log t}]$.
        \item[(2)] $\Pr_\br[\exists k \in W: \br \cdot \bv_{ik}\geq \frac{9-3\alpha_k-6\sqrt{(1 - \alpha_k^2)c}}{4\sqrt{1 - (\frac{1}{4} + \frac{3}{4}\alpha_k)^2}}t - O(\log t)] \geq \frac{1}{200(\log\log t)^2}$.
        \item[(3)] For every $k \in W$, let $V_{ik} = \{\ell \in N(k) \mid -\frac{3\sqrt{3}}{4}\cdot \frac{c}{1 + c} - \frac{C}{\log t}\leq \bv_{k\ell} \cdot \bv_{i} \leq \frac{\sqrt{3}}{2}\cdot \frac{c}{1+c} + \frac{C}{\log t}\}$, then the set $\{\bv_{k\ell} \mid \ell \in V_{ik}\}$ is a $(\sqrt{3}t, \frac{1}{4 (\log\log t)^2})$-cover.
    \end{itemize}
    Here $C > 0$ is a constant depending only on $c$. 
\end{theorem}
\begin{proof}
    Let $i \in V'$ be the vertex satisfying \eqref{eq:good_i}. Let $S = \left\{ j \in N(i) \mid \nu(i, j) \geq \frac{\mu_i(j)}{8(\log \log t)^2} \right\}$. By~\eqref{eq:good_i} we have $\mu_i(S) \geq \frac{1}{100(\log \log t)^2}$. Let us define 
    \begin{equation}
        W \coloneqq \{k \in N^{(2)}(i) \mid \exists j \in S, k \text{ good for } (i, j)\}.
    \end{equation}
    We prove that this choice of $W$ satisfies the conditions in our theorem. Items (1) and (2) are immediately given by Theorem~\ref{thm:two_steps_main}.

    It remains to show item (3). To avoid clutter, let us ignore $o(1)$ terms in the inner product which do not affect the argument meaningfully. Suppose $k \in W$. Then there exists some $j \in N(i)$ such that $k$ is good for $(i, j)$. By definition, we have $\bv_{ji} \cdot \bv_{jk} \in [0, \frac{c}{1+c}]$. Let
    \begin{equation}
        U_k = \left\{\ell \in N(k) \mid \bv_{k\ell} \cdot \bv_{kj} \in \left[0, \frac{c}{1+c}\right] \text{ and } \bv_{k\ell} \cdot \bv_{ji} \in \left[-\frac{c}{1+c}, \frac{c}{1+c}\right]\right\}.
    \end{equation}
    By~\eqref{eq:k_good_def}, we have $\mu_k(U_k)\geq \frac{1}{4 (\log\log t)^2}$. Note that for every $\ell \in N(k)$ we have
    \begin{align}
        \bv_{k\ell} \cdot \bv_i & = \bv_{k\ell} \cdot \left(\frac{1}{4}\bv_k - \frac{\sqrt{3}}{4}\bv_{kj} + \frac{\sqrt{3}}{2}\bv_{ji}\right) \\
        & = - \frac{\sqrt{3}}{4} \bv_{k\ell}\cdot \bv_{kj} + \frac{\sqrt{3}}{2} \bv_{k\ell} \cdot \bv_{ji}.
    \end{align}
    So for every $\ell \in U_k$, we have 
    \begin{equation}\label{eq:gamma_kell}
    -\frac{3\sqrt{3}}{4}\cdot \frac{c}{1 + c} \leq \bv_{k\ell} \cdot \bv_{i} \leq \frac{\sqrt{3}}{2}\cdot \frac{c}{1+c}.
    \end{equation}
    It follows that $\ell \in V_{ik}$, and therefore $U_k \subseteq V_{ik}$ and $\mu_k(V_{ik}) \geq \mu_k(U_k) \geq  \frac{1}{4 (\log\log t)^2}$. So $\{\bv_{k\ell} \mid \ell \in V_{ik}\}$ is a $(\sqrt{3}t, \frac{1}{4 (\log\log t)^2})$-cover as desired.
\end{proof}

\subsection{Proof of the Main Theorem}\label{subsec:win-win}

Theorem~\ref{thm:our_pruning_main} allows us to go one step beyond the second-level neighborhood studied in Theorem~\ref{thm:chlamtac_pruning_result}. We now analyze the gain that we obtain from this extra step. We first compose the two-step cover from Theorem~\ref{thm:chlamtac_pruning_result} with the third step.

\begin{lemma}\label{lem:third_level_composition}
    Let $i$, $W$, and $\{V_{ik} \mid k \in W\}$ be as in Theorem~\ref{thm:our_pruning_main}. Let $\eta > 0$ satisfy $\eta = \Theta(1)$ and 
    \[
    \eta \leq \min_{k \in W}\frac{9-3\alpha_k-6\sqrt{(1 - \alpha_k^2)c}}{4\sqrt{1 - (\frac{1}{4} + \frac{3}{4}\alpha_k)^2}} - C_0\frac{\log t}{t},
    \] where $C_0 > 0$ is sufficiently large (depending only on $c$) to absorb the $O(\log t)$ loss in Theorem~\ref{thm:our_pruning_main}(2). For every $\ell \in V_{ik}$, define $\bu_{k\ell} = \bv_{k\ell} - (\bv_{k\ell} \cdot \bv_{ki})\bv_{ki}$. Assume that $|W| \leq (\tail(\eta t))^{-(1 + c')}$. Then
    \begin{equation}
        \Pr[\exists k \in W, \ell \in V_{ik}: \br \cdot \bv_{ik} \leq -\eta t,\, \br \cdot \bu_{k\ell} \geq \sqrt{3}t - \sqrt{c'}\|\bu_{k\ell}\| \eta t - O(\log t)] = \Omega\left(\frac{1}{(\log \log t)^2}\right).
    \end{equation}
\end{lemma}
\begin{proof}
    By item (2) of Theorem~\ref{thm:our_pruning_main}, $\{ \bv_{ik} \mid  k \in W \}$ is an $(\eta t, \frac{1}{200(\log \log t)^2})$-cover. Since $|W| \leq (\tail(\eta t))^{-(1 + c')}$, $W$ is (at most) $c'$-inefficient.
    
    For every $k \in W$, by item (3) in Theorem~\ref{thm:our_pruning_main}, we have
    \begin{equation}
        \Pr_\br[\exists \ell \in V_{ik}: \br \cdot \bv_{k\ell} \geq \sqrt{3}t] \geq \frac{1}{4(\log \log t)^2}.
    \end{equation}

    Note that since $\bv_{k\ell}$ and $\bv_{ki}$ are both perpendicular to $\bv_k$, we have that $\bu_{k\ell} = \bv_{k\ell} - (\bv_{k\ell} \cdot \bv_{ki})\bv_{ki}$ is perpendicular to the vector space spanned by $\bv_k$ and $\bv_{ki}$, and in particular $\bu_{k\ell} \perp \bv_{ik}$. We can therefore apply Lemma~\ref{lem:cover_composition} (and Lemma~\ref{lem:v_i-v_0}) to $\{-\bv_{ik} \mid k \in W\}$ and $\{\{\bu_{k\ell} \mid \ell \in V_{ik}\} \mid k \in W\}$, with $s_1 = \eta t$ and $s_2 = \sqrt{3}t$, from which this lemma follows.
\end{proof}

We are now ready to prove our main statement. Recall that the functions $f_{c'}$ and $g_{c'}$ are defined as
   \[
        f_{c'}(n, \eta) = \Delta^{\frac{\eta^2(1 + c')}{3(1+c)}}, \quad g_{c'}(n, \lambda) =  \Delta^{\frac{\lambda^2}{3(1+c)} -\frac{1+3c}{5-c}}
    \]
    where 
    \[
        \eta_0 = \eta_0(c) = \inf_{\alpha \in [0, c/(1+c)]}\frac{9-3\alpha-6\sqrt{(1 - \alpha^2)c}}{4\sqrt{1 - \left(\frac{1}{4} + \frac{3}{4}\alpha\right)^2}},
    \]
    and
    \[
    \lambda_0 = \lambda_0(c, c') = \inf_{\substack{\beta \in [\frac{1}{4}, \frac{1}{4} + \frac{3c}{4(1+c)}] \\ \gamma \in [-\frac{3\sqrt{3}c}{4(1+c)}, \frac{\sqrt{3}c}{2(1+c)}]}} \frac{\left(\frac{1}{2}\sqrt{1 - \beta^2} + \frac{\sqrt{3}}{2}\cdot\frac{\beta}{\sqrt{1 - \beta^2}} \gamma - \frac{\sqrt{3}}{2}\sqrt{c'}\sqrt{1 - \frac{\gamma^2}{1 - \beta^2}}\right) \eta_0 +  \frac{3}{2} }{\sqrt{1 - \left(-\frac{1}{2}\beta + \frac{\sqrt{3}}{2}\gamma\right)^2}}.
    \]
\begin{theorem}[Theorem~\ref{thm:our_clean_statement} restated]\label{thm:third_step_main}
    Fix $0<c<1/5$, and let $n$ be sufficiently large. Assume that $\KMSp$ fails on some $n$-vertex graph $G$ with $c$-inefficient parameter $t$. Then, for any $c' > 0$, $\eta \in (0, \eta_0)$ and $\lambda \in (0, \lambda_0)$, we can find an independent set of size \[\Omega(\min(f_{c'}(n, \eta),g_{c'}(n, \lambda))).
    \] 
\end{theorem}
\begin{proof}
Let $i$, $W$, and $\{V_{ik} \mid k \in W\}$ be as in Theorem~\ref{thm:our_pruning_main}. 
By item (2) in Theorem~\ref{thm:our_pruning_main}, $\{ \bv_{ik} \mid  k \in W \}$ is a $(\eta_0 \cdot t - o(t), \Omega(\frac{1}{(\log \log t)^2}))$-cover. Here and below, the $o(t)$ loss is chosen to absorb both the error from the range of $\alpha_k$ and the $C_0\log t$ loss in Lemma~\ref{lem:third_level_composition}, so that the lemma's threshold can be taken to be $\eta_0 t - o(t)$. We have the following two cases:
\begin{itemize}
    \item $|W| \geq (\tail(\eta_0 \cdot t - o(t)))^{-(1 + c')}$. Then by using the vector $(2 + o(1))$-coloring from Lemma~\ref{lemma:second_level_vector_coloring}, we may find an independent set within $W$ whose size is $\widetilde{\Omega}(|W|)$. This corresponds to the function $f_{c'}$.
    \item Otherwise, $|W| \leq (\tail(\eta_0 \cdot t - o(t)))^{-(1 + c')}$, or in other words, $|W|$ is at most $c'$-inefficient, in which case we apply Lemma~\ref{lem:third_level_composition}. This case will correspond to the function $g_{c'}$.
\end{itemize}  

We now analyze the independent set size obtained in the second case. To this end, we analyze the quality of the third-step cover as follows. For any $k \in W$ and $\ell \in V_{ik}$, let us write $\beta_k = \frac{1}{4} + \frac{3}{4}\alpha_k = \bv_i \cdot \bv_k$, $\gamma_{k\ell} = \bv_i \cdot \bv_{k\ell}$. Then, (ignoring $o(1)$ terms for simplicity) we have $\beta_k \in [\frac{1}{4}, \frac{1}{4} + \frac{3}{4} \frac{c}{1+c}]$ and $\gamma_{k\ell} \in [-\frac{3\sqrt{3}}{4}\frac{c}{1+c}, \frac{\sqrt{3}}{2}\frac{c}{1+c}]$ (see \eqref{eq:gamma_kell}). We can decompose $\bv_\ell$ as follows:
\begin{align}
    \bv_\ell & = -\frac{1}{2} \bv_k + \frac{\sqrt{3}}{2}\bv_{k\ell} \\
    & =  -\frac{1}{2}\left( \beta_k \bv_i + \sqrt{1 - \beta_k^2}\bv_{ik}\right) + \frac{\sqrt{3}}{2}\bv_{k\ell} \\
    & = \left(-\frac{1}{2}\beta_k + \frac{\sqrt{3}}{2}\gamma_{k\ell}\right)\bv_i + \left(-\frac{1}{2}\sqrt{1 - \beta_k^2} - \frac{\sqrt{3}}{2}\cdot\frac{\beta_k}{\sqrt{1 - \beta_k^2}} \gamma_{k\ell}\right)\bv_{ik} \notag\\*
    &\qquad + \frac{\sqrt{3}}{2}\left(\bv_{k\ell} - \gamma_{k\ell}\bv_i + \frac{\beta_k \gamma_{k\ell}}{\sqrt{1 - \beta_k^2}}\bv_{ik}\right) \\
    & = \left(-\frac{1}{2}\beta_k + \frac{\sqrt{3}}{2}\gamma_{k\ell}\right)\bv_i + \left(-\frac{1}{2}\sqrt{1 - \beta_k^2} - \frac{\sqrt{3}}{2}\cdot\frac{\beta_k}{\sqrt{1 - \beta_k^2}} \gamma_{k\ell}\right)\bv_{ik} \notag\\*
    &\qquad + \frac{\sqrt{3}}{2}\left(\bv_{k\ell} - \frac{\gamma_{k\ell}}{\sqrt{1 - \beta_k^2}}\bv_{ki}\right)
\end{align}
Here, in the last equality, we used Fact~\ref{fact:vij_vji}. Note that
\begin{align*}
    \bv_{k\ell} \cdot \sqrt{1 - \beta_k^2}\bv_{ki} = \bv_{k\ell} \cdot (\beta_k \bv_k + \sqrt{1 - \beta_k^2}\bv_{ki}) = \gamma_{k\ell},
\end{align*}
which implies $(\bv_{k\ell} - \frac{\gamma_{k\ell}}{\sqrt{1 - \beta_k^2}}\bv_{ki}) \perp \bv_{ki}$. Consequently,
$\bv_i, \bv_{ik}$, and $(\bv_{k\ell} - \frac{\gamma_{k\ell}}{\sqrt{1 - \beta_k^2}}\bv_{ki})$ are pairwise orthogonal, so $-\frac{1}{2}\beta_k + \frac{\sqrt{3}}{2}\gamma_{k\ell} = \bv_\ell \cdot \bv_i$ is a value independent of $k$.
Since $\beta_k \geq \frac{1}{4}$ and $\gamma_{k\ell} \leq \frac{\sqrt{3}}{2}\frac{c}{1+c}$, we have $\bv_\ell \cdot \bv_i \leq -\frac{1}{8} + \frac{3c}{4(1+c)}$. By Lemma~\ref{lem:vector_coloring_negative}, we obtain a vector $\kappa$-coloring on $\cup_{k \in W}V_{ik}$ \footnote{Technically, a $(\kappa + o(1))$-coloring.}, where 
\begin{equation}
    \kappa = \frac{3 - 6\left(-\frac{1}{8} + \frac{3c}{4(1+c)}\right)}{1 - 4\left(-\frac{1}{8} + \frac{3c}{4(1+c)}\right)} = \frac{15 - 18\cdot \frac{c}{1+c}}{6 - 12\cdot \frac{c}{1+c}} = \frac{5 - c}{2 - 2c}.
\end{equation}
Then by applying Theorem~\ref{theorem:kms}, we obtain an independent set of size $\widetilde{\Omega}(|\cup_{k \in W} V_{ik}| \cdot \Delta^{-\frac{1+3c}{5-c}})$. It remains to bound $|\cup_{k \in W} V_{ik}|$. By Lemma~\ref{lem:third_level_composition}, setting $\bu_{k\ell} = \bv_{k\ell} - \frac{\gamma_{k\ell}}{\sqrt{1 - \beta_k^2}}\bv_{ki}$, we have
\begin{equation}
    \Pr[\exists k \in W, \ell \in V_{ik}: \br \cdot \bv_{ik} \leq -\eta_0 t + o(t),\, \br \cdot \bu_{k\ell} \geq \sqrt{3}t - \sqrt{c'}\|\bu_{k\ell}\| \eta_0 \cdot t - o(t)] = \Omega\left(\frac{1}{(\log \log t)^2}\right).
\end{equation}

Note that when $\br \cdot \bv_{ik} \leq -\eta_0 t + o(t),\, \br \cdot \bu_{k\ell} \geq \sqrt{3}t - \sqrt{c'}\|\bu_{k\ell}\| \eta_0 \cdot t - o(t)$, we have 
\begin{align}
    \br \cdot (\bv_\ell - (\bv_i \cdot \bv_\ell) \bv_i) &= \br \cdot \left(\left(-\frac{1}{2}\sqrt{1 - \beta_k^2} - \frac{\sqrt{3}}{2}\cdot\frac{\beta_k}{\sqrt{1 - \beta_k^2}} \gamma_{k\ell}\right)\bv_{ik} + \frac{\sqrt{3}}{2}\bu_{k\ell}\right) \\
    & \geq \left(\frac{1}{2}\sqrt{1 - \beta_k^2} + \frac{\sqrt{3}}{2}\cdot\frac{\beta_k}{\sqrt{1 - \beta_k^2}} \gamma_{k\ell}\right) \eta_0 t +  \frac{\sqrt{3}}{2} \left( \sqrt{3}t - \sqrt{c'}\|\bu_{k\ell}\| \eta_0 t\right)  - o(t) \\
    & \geq \left(\frac{1}{2}\sqrt{1 - \beta_k^2} + \frac{\sqrt{3}}{2}\cdot\frac{\beta_k}{\sqrt{1 - \beta_k^2}} \gamma_{k\ell} - \frac{\sqrt{3}}{2}\sqrt{c'}\sqrt{1 - \frac{\gamma_{k\ell}^2}{1 - \beta_k^2}}\right) \eta_0 t +  \frac{3}{2}t - o(t)
\end{align}
So the vectors $\{\frac{\bv_\ell - (\bv_i \cdot \bv_\ell) \bv_i}{\|\bv_\ell - (\bv_i \cdot \bv_\ell) \bv_i\|} \mid \ell \in \cup_{k \in W} V_{ik}\}$ form a $(\lambda_0 \cdot t - o(t), \Omega\left(\frac{1}{(\log \log t)^2}\right))$-cover, where 
\begin{equation}
    \lambda_0 = \inf_{\substack{\beta \in [\frac{1}{4}, \frac{1}{4} + \frac{3c}{4(1+c)}] \\ \gamma \in [-\frac{3\sqrt{3}c}{4(1+c)}, \frac{\sqrt{3}c}{2(1+c)}]}} \frac{\left(\frac{1}{2}\sqrt{1 - \beta^2} + \frac{\sqrt{3}}{2}\cdot\frac{\beta}{\sqrt{1 - \beta^2}} \gamma - \frac{\sqrt{3}}{2}\sqrt{c'}\sqrt{1 - \frac{\gamma^2}{1 - \beta^2}}\right) \eta_0 +  \frac{3}{2} }{\sqrt{1 - \left(-\frac{1}{2}\beta + \frac{\sqrt{3}}{2}\gamma\right)^2}}
\end{equation}
Thus, it follows from Fact~\ref{fact:cover_lower_bd} that $|\cup_{k \in W} V_{ik}| = \widetilde{\Omega}(\tail(\lambda_0 \cdot t - o(t))^{-1})$, and therefore $|\cup_{k \in W} V_{ik}| = \widetilde{\Omega}\left(\Delta^{\frac{(\lambda_0 - o(1))^2}{3(1+c)}}\right)$.
\end{proof}

\begin{corollary}\label{cor:choosing_c'}
    Let $c=0.0393241$, let $n$ be sufficiently large, and set
    \[
        D=n^{\frac{3+3c}{5+3c}}.
    \]
    Let $G$ be an $n$-vertex graph with $\Delta(G)\leq D$,
    and choose $t$ so that
    $\tail(t)=D^{-1/(3(1+c))}$.
    If $\KMSp$ fails on $G$ with parameter $t$, then we can
    find an independent set in $G$ of size
    $\Omega(n^{0.8046102})$.
\end{corollary}

\begin{proof}
    The argument of Theorem~\ref{thm:third_step_main}
    applies with $D$ in place of $\Delta(G)$.
    Indeed, all neighborhood sizes are bounded by $D$,
    the chosen value of $\tail(t)$ expresses the Gaussian
    tail estimates in terms of $D$, and the final KMS
    degree factor satisfies
    \[
        \Delta(G)^{-\frac{1+3c}{5-c}}
        \geq D^{-\frac{1+3c}{5-c}}.
    \]
    Thus, for any admissible $c',\eta,\lambda$, the same
    argument gives an independent set of size
    \[
        \Omega\left(
            \min\left\{
                D^{\frac{\eta^2(1+c')}{3(1+c)}},
                D^{\frac{\lambda^2}{3(1+c)}
                   -\frac{1+3c}{5-c}}
            \right\}
        \right).
    \]

    Take
    \[
        c'=0.0258187,\qquad
        \eta=2.0035790,\qquad
        \lambda=2.1956003.
    \]
    These satisfy $\eta<\eta_0$ and $\lambda<\lambda_0$.
    Moreover,
    \[
        \frac{\eta^2(1+c')}{5+3c}>0.8046102,
    \]
    and
    \[
        \frac{\lambda^2}{5+3c}
        -\frac{3+3c}{5+3c}\frac{1+3c}{5-c}
        >0.8046102.
    \]
    Substituting the value of $D$ proves the claim.
\end{proof}

\begin{proof}[Proof of Theorem~\ref{theorem:main}]
    Set $c=0.0393241$ and
    $D=n^{(3+3c)/(5+3c)}$.
    Consider an $n$-vertex graph $G$ with $\Delta(G)\leq D$,
    and run $\KMSp$ with
    \[
        \tail(t)=D^{-1/(3(1+c))}.
    \]
    If $\KMSp$ does not fail, its expected output size is
    \[
        \Omega(n\tail(t))
        =\Omega\left(n^{1-\frac{1}{5+3c}}\right).
    \]
    Otherwise, Corollary~\ref{cor:choosing_c'} gives an
    independent set of size $\Omega(n^{0.8046102})$.
    Since $1-1/(5+3c)<0.8046102$, both cases make progress
    toward an $O(n^{1/(5+3c)})$-coloring.

    Finally, $1/(5+3c)<0.19539$, so
    Corollary~\ref{cor:sdp_combination}, with a slightly
    smaller constant than $c$, yields an
    $O(n^{0.19539})$-coloring.
\end{proof}

\section{Conclusion}\label{sec:conclusion}

In this paper, we presented an $O(n^{0.19539})$-coloring algorithm for 3-colorable graphs, improving the state of the art on SDP-based approaches. We conclude with some directions for future work.
\begin{itemize}
    \item The main loss term in the current analysis is incurred by the cover composition lemma (Lemma~\ref{lem:cover_composition}), where the analysis uses Gaussian isoperimetry, which is only tight for half-spaces. However, for covers the geometry seems very different from half-spaces which gives hope for improvement. Indeed,  Arora,  Chlamtáč, and Charikar~\cite{arora_new_2006} proposed the following conjecture, which, if true, would provide improved parameters in the cover composition lemma:
    \begin{conjecture}[Conjecture 1 in~\cite{arora_new_2006}]
        For any $(s, \delta)$-cover $\{\bu_i\}_{i = 1}^k$, if $\{\bv_i\}_{i = 1}^k$ are mutually equidistant unit vectors which also form an $(s, \delta)$-cover, then for any $\rho \geq 0$
        \[
        \Pr[\max_i \br \cdot \bu_i \leq s - \rho] \leq \Pr[\max_i \br \cdot \bv_i \leq s - \rho].
        \]
    \end{conjecture}
    \item The most natural next step of work is perhaps to extend the analysis further to four steps. We believe this can be done. However, without improving the loss term discussed in the previous item, the potential gain from this would likely be more technical trouble than it's worth.
    \item It is interesting to study whether our vector $5/2$-coloring may be applicable to the combinatorial coloring algorithms. One potential way to do so was suggested in Lemma~\ref{lem:combinatorial_vector_coloring}.
\end{itemize}

\bibliography{references_arxiv_v2}

\newcommand{\etalchar}[1]{$^{#1}$}
\begin{thebibliography}{KOW{\v{Z}}23}

\bibitem[ACC06]{arora_new_2006}
Sanjeev Arora, Eden Chlamtáč, and Moses Charikar.
\newblock New approximation guarantee for chromatic number.
\newblock In {\em Symposium on {Theory} of {Computing}}, {STOC}, pages
  215--224, 2006.

\bibitem[AG11]{AG11}
Sanjeev Arora and Rong Ge.
\newblock New tools for graph coloring.
\newblock In {\em Approximation, Randomization, and Combinatorial Optimization,
  {APPROX-RANDOM}}, pages 1--12, 2011.

\bibitem[AK97]{AK97}
Noga Alon and Nabil Kahale.
\newblock A spectral technique for coloring random 3-colorable graphs.
\newblock {\em SIAM Journal on Computing}, 26:1733--1748, 1997.

\bibitem[ARV09]{arora2009expander}
Sanjeev Arora, Satish Rao, and Umesh Vazirani.
\newblock Expander flows, geometric embeddings and graph partitioning.
\newblock {\em Journal of the ACM (JACM)}, 56(2):1--37, 2009.

\bibitem[BBKO21]{barto2021algebraic}
Libor Barto, Jakub Bul{\'\i}n, Andrei Krokhin, and Jakub Opr{\v{s}}al.
\newblock Algebraic approach to promise constraint satisfaction.
\newblock {\em Journal of the ACM (JACM)}, 68(4):1--66, 2021.

\bibitem[BHK25]{BHK25}
Mitali Bafna, Jun-Ting Hsieh, and Pravesh~K. Kothari.
\newblock Rounding large independent sets on expanders.
\newblock In {\em Symposium on Theory of Computing, {STOC}}, pages 631--642,
  2025.

\bibitem[BK97]{blum_tildeon314-coloring_1997}
Avrim Blum and David Karger.
\newblock An $\tilde{O}(n^{3/14})$-coloring algorithm for 3-colorable graphs.
\newblock {\em Information Processing Letters}, 61(1):49--53, 1997.

\bibitem[BKLM22]{braverman2022invariance}
Mark Braverman, Subhash Khot, Noam Lifshitz, and Dor Minzer.
\newblock An invariance principle for the multi-slice, with applications.
\newblock In {\em Foundations of Computer Science (FOCS)}, pages 228--236,
  2022.

\bibitem[Blu94]{blum_new_1994}
Avrim Blum.
\newblock New approximation algorithms for graph coloring.
\newblock {\em Journal of the ACM (JACM)}, 41(3):470--516, 1994.

\bibitem[Bor85]{borell1985geometric}
Christer Borell.
\newblock Geometric bounds on the {Ornstein-Uhlenbeck} velocity process.
\newblock {\em Probability Theory and Related Fields}, 70(1):1--13, 1985.

\bibitem[BR90]{berger1990better}
Bonnie Berger and John Rompel.
\newblock A better performance guarantee for approximate graph coloring.
\newblock {\em Algorithmica}, 5(1):459--466, 1990.

\bibitem[BS95]{BS95}
Avrim Blum and Joel Spencer.
\newblock Coloring random and semi-random k-colorable graphs.
\newblock {\em J. Algorithms}, 19(2):204--234, 1995.

\bibitem[Chl07]{chlamtac_approximation_2007}
Eden Chlamtáč.
\newblock Approximation {Algorithms} {Using} {Hierarchies} of {Semidefinite}
  {Programming} {Relaxations}.
\newblock In {\em {Foundations} of {Computer} {Science}, {FOCS}}, pages
  691--701, 2007.

\bibitem[Chl09]{chlamtac_non-local_2009}
Eden Chlamtáč.
\newblock {\em Non-{Local} {Analysis} of {SDP}-{Based} {Approximation}
  {Algorithms}}.
\newblock PhD thesis, Princeton University, January 2009.

\bibitem[C{\v{Z}}25]{ciardo2025approximate}
Lorenzo Ciardo and Stanislav {\v{Z}}ivn{\'y}.
\newblock Approximate graph coloring and the crystal with a hollow shadow.
\newblock {\em SIAM Journal on Computing}, 54(4):1138--1192, 2025.

\bibitem[DF16]{DF16}
Roee David and Uriel Feige.
\newblock On the effect of randomness on planted 3-coloring models.
\newblock In {\em Symposium on Theory of Computing, {STOC}}, pages 77--90,
  2016.

\bibitem[DHSV15]{dinur2014derandomized}
Irit Dinur, Prahladh Harsha, Srikanth Srinivasan, and Girish Varma.
\newblock Derandomized graph product results using the low degree long code.
\newblock In {\em Symposium on Theoretical Aspects of Computer Science,
  {STACS}}, pages 275--287, 2015.

\bibitem[DKK{\etalchar{+}}18]{dinur2018towards}
Irit Dinur, Subhash Khot, Guy Kindler, Dor Minzer, and Muli Safra.
\newblock Towards a proof of the 2-to-1 games conjecture?
\newblock In {\em Symposium on Theory of Computing, STOC}, pages 376--389,
  2018.

\bibitem[DMR09]{dinur_conditional_2009}
Irit Dinur, Elchanan Mossel, and Oded Regev.
\newblock Conditional {Hardness} for {Approximate} {Coloring}.
\newblock {\em SIAM Journal on Computing}, 39(3):843--873, 2009.

\bibitem[FK01]{FK01}
Uriel Feige and Joe Kilian.
\newblock Heuristics for semirandom graph problems.
\newblock {\em J. Comput. Syst. Sci.}, 63(4):639--671, 2001.

\bibitem[GJS74]{garey1974some}
Michael~R Garey, David~S Johnson, and Larry Stockmeyer.
\newblock Some simplified {NP}-complete problems.
\newblock In {\em Symposium on Theory of computing, STOC}, pages 47--63, 1974.

\bibitem[GK04]{guruswami2004hardness}
Venkatesan Guruswami and Sanjeev Khanna.
\newblock On the hardness of 4-coloring a 3-colorable graph.
\newblock {\em SIAM Journal on Discrete Mathematics}, 18(1):30--40, 2004.

\bibitem[GS20]{guruswami_d--1_2020}
Venkatesan Guruswami and Sai Sandeep.
\newblock d-{To}-1 {Hardness} of {Coloring} 3-{Colorable} {Graphs} with {O}(1)
  {Colors}.
\newblock In {\em {International} {Colloquium} on {Automata}, {Languages}, and
  {Programming}, {ICALP}}, pages 62:1--62:12, 2020.

\bibitem[Hsi26]{Hsieh26}
Jun-Ting Hsieh.
\newblock Coloring 3-colorable graphs with low threshold rank.
\newblock In {\em Symposium on Discrete Algorithms, SODA}, pages 718--726,
  2026.

\bibitem[KLS00]{khanna2000hardness}
Sanjeev Khanna, Nathan Linial, and Shmuel Safra.
\newblock On the hardness of approximating the chromatic number.
\newblock {\em Combinatorica}, 20(3):393--415, 2000.

\bibitem[KLT17]{KLT17}
Akash Kumar, Anand Louis, and Madhur Tulsiani.
\newblock Finding pseudorandom colorings of pseudorandom graphs.
\newblock In {\em Foundations of Software Technology and Theoretical Computer
  Science, {FSTTCS}}, pages 37:1--37:12, 2017.

\bibitem[KMS98]{karger_approximate_1998}
David Karger, Rajeev Motwani, and Madhu Sudan.
\newblock Approximate graph coloring by semidefinite programming.
\newblock {\em Journal of the ACM}, 45(2):246--265, March 1998.

\bibitem[KMS17]{khot2017independent}
Subhash Khot, Dor Minzer, and Muli Safra.
\newblock On independent sets, 2-to-2 games, and grassmann graphs.
\newblock In {\em Symposium on Theory of Computing, {STOC}}, pages 576--589,
  2017.

\bibitem[KMS23]{khot2023pseudorandom}
Subhash Khot, Dor Minzer, and Muli Safra.
\newblock Pseudorandom sets in grassmann graph have near-perfect expansion.
\newblock {\em Annals of Mathematics}, 198(1):1--92, 2023.

\bibitem[KOW{\v{Z}}23]{krokhin2023topology}
Andrei Krokhin, Jakub Opr{\v{s}}al, Marcin Wrochna, and Stanislav
  {\v{Z}}ivn{\'y}.
\newblock Topology and adjunction in promise constraint satisfaction.
\newblock {\em SIAM Journal on Computing}, 52(1):38--79, 2023.

\bibitem[KT12]{kawarabayashi2012combinatorial}
Ken-ichi Kawarabayashi and Mikkel Thorup.
\newblock Combinatorial coloring of 3-colorable graphs.
\newblock In {\em Foundations of Computer Science, FOCS}, pages 68--75, 2012.

\bibitem[KT17]{kawarabayashi_coloring_2017}
Ken-ichi Kawarabayashi and Mikkel Thorup.
\newblock Coloring 3-{Colorable} {Graphs} with {Less} than $n^{1/5}$ {Colors}.
\newblock {\em J. ACM}, 64(1):4:1--4:23, March 2017.

\bibitem[KTY24]{kawarabayashi_better_2024}
Ken-ichi Kawarabayashi, Mikkel Thorup, and Hirotaka Yoneda.
\newblock Better {Coloring} of 3-{Colorable} {Graphs}.
\newblock In {\em {Symposium} on {Theory} of {Computing}, {STOC}}, pages
  331--339, 2024.

\bibitem[Las01]{lasserre2001global}
Jean~B Lasserre.
\newblock Global optimization with polynomials and the problem of moments.
\newblock {\em SIAM Journal on optimization}, 11(3):796--817, 2001.

\bibitem[Par00]{parrilo2000structured}
Pablo~A Parrilo.
\newblock {\em Structured semidefinite programs and semialgebraic geometry
  methods in robustness and optimization}.
\newblock PhD thesis, California Institute of Technology, 2000.

\bibitem[Wig83]{wigderson_improving_1983}
Avi Wigderson.
\newblock Improving the performance guarantee for approximate graph coloring.
\newblock {\em Journal of the ACM, JACM}, 30(4):729--735, 1983.

\end{thebibliography}
\bibliographystyle{alpha}

\appendix

\section{Missing Proofs from Section~\ref{sec:two_steps}}

\subsection{Proof of Theorem~\ref{thm:chlamtac_pruning_result}}\label{subsec:thm_chlamtac_pruning_result}

We state a simple result that we'll use a few times in the proof.

\begin{lemma}\label{lem:pruning}
    Let $G = (V, E, w)$ be an $n$-vertex undirected weighted graph. Assume there exists some $r > 0$ such that $\sum_{e \in E} w(e) \geq rn$, then there exists a nonempty induced subgraph $G' = (V', E', w|_{E'})$ such that for every $i \in V'$, $\sum_{j: \{i, j\} \in E'}w(\{i, j\}) \geq r$. Moreover, $\sum_{e\in E'}w(e)\geq r|V'|$.
\end{lemma}
\begin{proof}
    If there exists $i \in V$ whose incident edges have total weight less than $r$, then we remove it and its incident edges from the graph. We repeat this process until there is no such $i$. Clearly, at each step, the weight of the edges we remove is less than $r$, and the process can be repeated for at most $n$ times, so the total weight of removed edges is less than $rn$. This means there must be edges that are not removed, and therefore we obtain a nonempty induced subgraph $G' = (V', E', w|_{E'})$ such that for every $i \in V'$, $\sum_{j: \{i, j\} \in E'}w(\{i, j\}) \geq r$. Since each deleted vertex removes less than $r$ edge weight, we also have
\[
    \sum_{e\in E'}w(e)
    \geq \sum_{e\in E}w(e)-r(n-|V'|)
    \geq r|V'|.
\]
\end{proof}

\begin{proof}[Proof of Theorem~\ref{thm:chlamtac_pruning_result}]
    Since $p_i \geq 1/2$ for at least $n/2$ vertices in $V$, we have $\sum_{e \in E} w(e) \geq n/8$. By Lemma~\ref{lem:pruning}, we may find a nonempty induced subgraph $G_1= (V_1, E_1)$ such that $\mu_i(N_{G_1}(i)) \geq 1/8$ for every $i \in V_1$. Moreover, $\sum_{e\in E_1}w(e)\geq |V_1|/8$. 
    
    By Theorem~\ref{thm:pruning_boosted}, for every $i \in V_1$ there exists $X_i \subseteq N_{G_1}(i)$ such that $\mu_i(X_i)\geq \mu_i(N_{G_1}(i))-\frac{1}{\log t}$ and $(X_i, \mu_i)$ is $(\frac{c}{1+c}+o(1), \exp(-\Omega(\log^2 t)))$-spread. We now remove every edge $\{i, j\} \in E_1$ for which $j \notin X_i$ or $i \notin X_j$. The total weight of deleted edges is at most
\[
    \sum_{i\in V_1}
    \mu_i\bigl(N_{G_1}(i)\setminus X_i\bigr)
    \leq \frac{|V_1|}{\log t}.
\]
Thus the remaining graph on $V_1$ has total edge weight at least
\[
    |V_1|\left(\frac18-\frac{1}{\log t}\right).
\] Applying Lemma~\ref{lem:pruning} again with
$r=1/8-1/\log t$, we obtain $G_2= (V_2, E_2)$ such that for every $i \in V_2$, $\mu_i(N_{G_2}(i)) \geq 1/8 - O(1/\log t)$ and $(N_{G_2}(i), \mu_i)$ is $(\frac{c}{1+c}+o(1), \exp(-\Omega(\log^2 t)))$-spread.
    
Let $L=\log\log t$ and set
\[
    r=\frac18-\frac1{\log t}-\frac{L+1}{L^2},
\]
which is positive for sufficiently large $t$.
By the additional conclusion of Lemma~\ref{lem:pruning},
\[
    \sum_{e\in E_2}w(e)
    \geq \left(\frac18-\frac1{\log t}\right)|V_2|.
\]

Starting from $H=G_2$, repeatedly perform either of the
following operations whenever applicable:
\begin{itemize}
    \item Delete a vertex whose weighted degree in $H$
    is less than $r$, together with its incident edges.
    \item For an ordered pair of adjacent vertices $(i,j)$,
    define
    \[
        E_{ij}(H)=
        \{k\in N_H(j):
          \bv_{ji}\cdot\bv_{jk}\geq-1/L\}.
    \]
    If $\mu_j(E_{ij}(H))<1/L^2$, delete every edge
    $\{j,k\}$ with $k\in E_{ij}(H)$.
\end{itemize}

For each $j\in V_2$, let $i_1,\ldots,i_q$ be the neighbors
that trigger the second operation at $j$. Their vectors
have pairwise inner products less than $-1/L$, since
otherwise a later triggering edge would already have
been deleted. Thus
\[
    0\leq
    \left\|\sum_{p=1}^q\bv_{ji_p}\right\|^2
    \leq q-\frac{q(q-1)}{L},
\]
so $q\leq L+1$. Consequently, the total edge weight
deleted by operations of the second type is at most
$|V_2|(L+1)/L^2$.

If every vertex were deleted, operations of the first
type would delete less than $r|V_2|$ edge weight.
The total deleted weight would therefore be less than
\[
    \left(r+\frac{L+1}{L^2}\right)|V_2|
    =\left(\frac18-\frac1{\log t}\right)|V_2|,
\]
contradicting the lower bound on the weight of $E_2$.
Hence the process terminates at a nonempty graph $G'$.

Every vertex of $G'$ has weighted degree at least
$r=1/8-O(1/\log\log t)$, and every remaining ordered
edge $(i,j)$ satisfies $\mu_j(E_{ij}(G'))\geq1/L^2$.
The inherited spread property gives
\[
    \mu_j(W_{ij})
    \geq \frac1{L^2}-\exp(-\Omega(\log^2t))
    \geq \frac1{2L^2}
\]
for sufficiently large $t$. The packing condition is
preserved under restriction, so $G'$ has the desired
properties.
\end{proof}

\end{document}